\documentclass{lmcs} %%% last changed 2014-08-20
\pdfoutput=1

\usepackage{lastpage}
\lmcsdoi{18}{2}{13}
\lmcsheading{}{\pageref{LastPage}}{}{}%
{Dec.~09,~2020}{Jun.~01,~2022}{}


\usepackage{amssymb}
\usepackage{graphicx}
\usepackage{alltt}
\usepackage{multicol}
\usepackage{amsmath}
\usepackage{amsfonts}
\usepackage{listings}
\usepackage{multicol}
\usepackage{mathtools}
\usepackage{float}
\usepackage{syntax}
\usepackage{ragged2e}
\usepackage{fancybox}
\usepackage{cases}
\usepackage{multirow}
\usepackage{tikz}
\usetikzlibrary{shapes, arrows, chains}
\usepackage{hyperref}
\usepackage[normalem]{ulem}
\usepackage{array}
\usepackage{tabularx}
\usepackage{subcaption} 
\usepackage[utf8]{inputenc}
\theoremstyle{plain} %\crefname{satz}{Satz}{S\"atze}
\theoremstyle{plain}\newtheorem{example}[thm]{Example} %\crefname{satz}{Satz}{S\"atze}
\usepackage{longtable}
\usepackage{amsthm,amsmath,amssymb,amscd,dsfont} 
\usepackage{comment}

\def\eg{{\em e.g.}}

\definecolor{darkGreen}{RGB}{20,150,50}
\definecolor{blue-violet}{rgb}{0.54, 0.17, 0.89}
\definecolor{pink}{RGB}{255,153,204}
\definecolor{darkBlue}{RGB}{10,40,160}

\newcommand{\JH}[1]{{\color{black} #1}}

\newcommand{\FG}[1]{{\color{black} #1}}
\newcommand{\MZ}[1]{{\color{black} #1}}
\newcommand{\EK}[1]{{\color{black} #1}}
\newcommand{\rev}[1]{{\color{black} #1}}

\newcommand{\overto}[1]{\stackrel{#1}{%
		\overrightarrow{\smash{\,\scriptsize{\phantom{#1}}\,}}}}

\def\bag{\mathit{bag}}
\def\now{\mathit{now}}
\def\arrival{\mathit{arrival}}
\def\deadline{\mathit{deadline}}
\def\after{\mathit{after}}
\def\msgsig{\mathit{msgsig}}
\def\statevars{\mathit{statevars}}
\def\Msg{\mathit{Msg}}
\def\ID{\mathit{ID}}
\def\Var{\mathit{Var}}
\def\Interval{\mathit{Var}_\Delta}
\def\Value{\mathit{Value}}
\def\life{\mathit{life}}
\def\RcvPublish{\mathit{RcvPublish}}

\def\lastPub{\mathit{lastPub}}

\def\consume{\mathit{consume}}
\def\consumed{\mathit{consumed}}

\def\assertion{\mathit{Assertion}}
\def\service{\mathit{service}}
\def\transmission{\mathit{transmissionTime}}

\def\latency{\mathit{LatencyOverLoad}}
\def\int{\mathbb{N}}

\lstdefinestyle{customjava}{
	belowcaptionskip=1\baselineskip,
	breaklines=true,
	numbers=left,
	xleftmargin=\parindent,
	basicstyle=\footnotesize\ttfamily,
	numberstyle=\color{black}
}

\newcommand{\powerset}[1]{\mathcal{P}(#1)}
\newcommand{\powermultiset}[1]{\mathcal{P}^{*}(#1)}
\newcommand{\longhooktrans}[1]{\,{\lhook\joinrel\xrightarrow{#1}}\,}
\newcommand{\longtrans}[1]{\xrightarrow{#1}}

\begin{document}
	
	\title[Specification and Verification of Timing Properties]{Specification and Verification of Timing Properties in Interoperable Medical Systems
		%Formal Analysis of QoS Requirements of Medical Systems
	}
	%\titlerunning{Specification and Verification of Timing Properties in Medical Systems}
	
	\author[M.~Zarneshan]{Mahsa Zarneshan\rsuper{a}}%[a]	%required
	%\address{School of Electrical and Computer Engineering, University of Tehran,Tehran, Iran}
	%\email{m.zarneshan@ut.ac.ir}  %optional
	
	\author[F.~Ghassemi]{Fatemeh Ghassemi\rsuper{a,b}\texorpdfstring{\textsuperscript{\textdagger}}{}}%[a,b]	%optional
	\thanks{\textsuperscript{\textdagger}{} Corresponding Author}
	%\address{School of Electrical and Computer Engineering, University of Tehran,Tehran, Iran , School of Computer Science, Institute for Research in Fundamental Sciences, PO. Box 19395-5746, Tehran, Iran}	%optional
	%\email{fghassemi@ut.ac.ir}
	
	\author[E.~Khamespanah]{Ehsan Khamespanah\rsuper{a}}%[a]	%optional
	%\address{School of Electrical and Computer Engineering, University of Tehran,		Tehran, Iran}	%optional
	%\email{e.khamespanah@ut.ac.ir}  %optional
	
	\author[M.~Sirjani]{Marjan Sirjani\rsuper{c}}%[c]	%optional
	%\address{School of Innovation, Design and Engineering, M\"{a}lardalen University, V\"{a}ster{\aa}s, Sweden}
%\email{marjan.sirjani@mdh.se}
		%optional
	
	\author[J.Hatcliff]{John Hatcliff\rsuper{d}}%[d]	%optional
%	\address{Computer Science Department, Kansas State University, USA}
%\email{hatcliff@ksu.edu}  %optional

	\address{School of Electrical and Computer Engineering, University of Tehran,Tehran, Iran}
	\email{m.zarneshan@ut.ac.ir,fghassemi@ut.ac.ir,e.khamespanah@ut.ac.ir}  %optional

	\address{School of Computer Science, Institute for Research in Fundamental Sciences, PO. Box 19395-5746, Tehran, Iran}

	\address{School of Innovation, Design and Engineering, M\"{a}lardalen University, V\"{a}ster{\aa}s, Sweden}
    \email{marjan.sirjani@mdh.se}

	\address{Computer Science Department, Kansas State University, USA}
	\email{hatcliff@ksu.edu}  %optional

	%% etc.
	
	%% required for running head on odd and even pages, use suitable
	%% abbreviations in case of long titles and many authors:
	
	%%%%%%%%%%%%%%%%%%%%%%%%%%%%%%%%%%%%%%%%%%%%%%%%%%%%%%%%%%%%%%%%%%%%%%%%%%%
	
	%% the abstract has to PRECEDE the command \maketitle:
	%% be sure not to issue the \maketitle command twice!
	
	\begin{abstract}
		To support the dynamic composition of various devices/apps into a medical system at point-of-care, a set of communication patterns to describe the communication needs of devices has been proposed. To address timing requirements, each pattern breaks common timing properties into finer ones that can be enforced locally by the components. Common timing requirements for the underlying communication substrate are derived from these local properties. The local properties of devices are assured by the vendors at the development time. Although organizations procure devices that are compatible in terms of their local properties and middleware, they may not operate as desired. The latency of the organization network interacts with the local properties of devices. To validate the interaction among the timing properties of components and the network, we formally specify such systems in Timed Rebeca. We use model checking to verify the derived timing requirements of the communication substrate in terms of the network and device models. We provide a set of templates as a guideline to specify medical systems in terms of the formal model of patterns. A composite medical system using several devices is subject to state-space explosion. We extend the reduction technique of Timed Rebeca based on the static properties of patterns. We prove that our reduction is sound and show the applicability of our approach in reducing the state space by modeling two clinical scenarios made of several instances of patterns.
	\end{abstract}
	\maketitle
	
	%% start the paper here:
	
	\section{Introduction}\label{S:intro}
	
	%\JH{Text from John}
	
	Medical Application Platforms (MAPs) \cite{hatcliff2012rationale} support the deployment of medical systems that are composed of medical devices and apps. The devices, apps, and the platform itself may be developed independently by different vendors.  
	The ASTM F2761 standard \cite{ASTM} specifies a
	particular MAP architecture called an Integrated Clinical Environment (ICE).
	The AAMI-UL 2800 standard complements F2671 by defining safety/security requirements of interoperable medical systems, including those built using the ICE architecture. Other medical domain standards such as the IEC 80001 series address safety, security, and risk management of medical information technology (IT) networks. 
	All of these standards, as well as emerging regulatory guidance documents for interoperable medical devices, emphasize the importance of accurately specifying the device and app interfaces, understanding the interactions between devices and apps, as well as the implications of those interactions (and associated failures) for safety and security. Timeliness is an important safety aspect of these interactions -- sensed information and actuation control commands need to be communicated to and from medical functions within certain latency bounds.
	A number of approaches to medical device interfacing have been proposed. Given the needs described above, interfacing approaches and interoperability platforms would clearly benefit from specification and verification frameworks that can define the interface capabilities of devices/apps and provide an automated means for verifying properties of interactions of devices/apps as they are composed into systems.
	% As advised by these standards, a model that specifies the capabilities and interface of devices/apps is essential. 
	Using such a framework that supports a compositional approach, we can inspect the substitutability and compatibility of devices to have a flexible and correct composite medical system. 
	
	% A device model together with its communication needs has been described in \cite{requirements}. 
	% From JH: I suggest that we reword the above because "model" in that paper is being used differently than "model" in this paper.
	
	Goals for a device interfacing framework, together with its communication needs have been described in \cite{requirements}.
	Working from these goals, a collection of communication patterns for MAPs was proposed in \cite{7318707} that can be implemented
	on widely available middleware frameworks. %\sout{A set of communication requirements that enables dynamic composition of devices and apps has been identified \cite{requirements}. As a solution, a set of communication patterns has been proposed in \cite{7318707} that can serve as the schema to describe the communication needs of devices/apps. These communication patterns, based on asynchronous message-passing, facilitate development and forensic analysis of clinical scenarios. The use of message passing as the basic communication model is quite common in Internet of Things applications. While the individual components can be very different and operate independently, their interactions typically expose and deliver important emergent properties \cite{hatcliff2012rationale}.}}
	These communication patterns address timing properties for medical systems built on a MAP. 
	% From JH: remove since the notion of an interoperable medical system and MAP has been defined earlier.
	% A medical system is composed of 
	% communicating devices/apps
	% (the so called components) that communicate through a communication substrate. 
	The patterns break down the timing properties into finer properties that can be locally monitored by each component.
	%\sout{consist of a set of components which are responsible to check a set of quality of service (QoS) properties locally. The combination of these quality of service properties should guarantee point-to-point communication requirements.} 
	The timing properties impose constraints on timing behavior of components like the minimum and maximum amount of time between consequent sent or handled messages. These constraints    balance the message passing 
	speed among components, and assure freshness of data. 
	The timing requirements of the communication substrate can be derived from these local timing properties. 
	The timing requirements of the communication substrate impose upper bounds on 
	communication latency. 
	With certain assumptions on the local timing properties, and configuring the network properly based on the derived requirements, we can guarantee the timing properties of the composite system.
	%The devices indeed satisfy the local QoS properties in conjunction with a specified middleware by construction at the development time. 
	%In a composite medical system there is no control over the local QoS properties. 
	Assuming that only devices compatible in terms of timing properties are composed together, they 
	may fail to operate as desired due to the interaction of network latency and local timing properties. Communication failures or unpredicted and undesired delays in medical systems may result in loss of life. 
	%  I would rather not say this so briefly here without providing context of the applications.
	% For example, the X-ray machine should stop after two seconds, otherwise it causes harmful prolonged exposure. 
	For example, considering some of the MAP applications outlined in \cite{ASTM} and \cite{hatcliff2012rationale},  
	in a scenario where a Patient-Controlled Analgesic (PCA) Pump is being controlled by a monitoring app, once the app receives data from patient monitoring device indicating that a patient's health is deteriorating, the app needs to send a halt command to the pump within a certain time bound to stop the flow of an opioid into the patient. In a scenario, where an app is pausing a ventilator to achieve a higher-quality x-ray, the ventilator needs to be restarted within a certain time bound.
	
	We can verify the satisfaction of timing communication requirements with regard to the network behavior and configuration before deployment. The verification results are helpful for dynamic network configuration or capacity planning. We assume that components in systems (both apps and devices) satisfy their timing constraints, checked using conventional timing analysis techniques. We focus on timing issues in the communication substrate. Components have no direct control over the communication substrate performance, and ensuring that the system performs correctly under varying network performance is a key concern of the system integrator.%\sout{A medical system may use several instances of such patterns among its constituent devices and apps. Adjusting these local properties is non-trivial and depends not only to the architecture of the system but also the underlying network.}
	
	We exploit model checking to verify that the configured devices together with assumptions about latencies in the deployed network 
	ensure timing requirements of medical systems before deployment. Each timing requirement expresses a requirement on the communication substrate for each involved pattern in the system. Each timing requirement of communication substrate imposes an upper bound on (logical) end-to-end communication latency between two components. %The communication requirement for each pattern expresses that the communication latency between two components does not exceed a constant value (derived based on the local QoS properties).} 
	When a component takes part in more than one pattern simultaneously, it will receive an interleaving of messages. The upper bound of the end-to-end communication latency depends on such interleaving. Model checking technique is a suitable approach that considers all possible interleaving of messages to verify the properties. We use the actor-based modeling language Rebeca \cite{sirjani2004modeling,Sirjani06} to verify the configuration of medical systems. %Actor model is a computational model for event-based distributed systems in which actors communicate by asynchronous message-passing. {\color{red}Implementations of the proposed pattern has been carried on two different frameworks in Java to validate the validity and evaluate the performance of patterns. These implementations are based on a layered architecture in which a layer is considered to encapsulate the details pertaining to communication, i.e., checking the local Qos properties. This layer works asynchronous with the layer as an interface to the communication substrate. \fixme{implementation of requester-responder} \sout{The computation model of Rebeca helps to model the communication patterns with minimal effort and mistake.}} 
	We exploit the timed extension of Rebeca to address local timing properties defined in terms of the timing behavior of components. Timed Rebeca \cite{TimedRebeca,SirjaniK16} is supported by the Afra tool suite which efficiently verifies timed properties by model checking.  %
	%Timed Rebeca's constructs %supports inheritance among actors which 
	%facilitates modeling of 
	We model communication patterns such that their components communicate over a shared communication substrate. %, and \FG{we can consider the impact of network sharing on communication latency} % by modeling the scheduling policy of the communication substrate on transferring messages.
	We provide a template for the shared communication substrate in Timed Rebeca, and it can be reused in modeling different medical systems irrespective of the number of components involved.
	
	%TODO: why Rebeca
	
	In this paper we model and analyze communication patterns in Timed Rebeca using the %implementation 
	architecture proposed for the communication patterns in the extended version of \cite{7318707}.
	For each pattern, this architecture considers an interface component on either side of the communication to abstract the lower-level details of the communication substrate. These interface components monitor the local timing properties of patterns. So for modeling devices/apps, we only focus on their logic for communicating messages through the interfaces of patterns reusing the proposed models of patterns. The interface components of patterns together with the communication substrate are modeled by distinct actors. Since the timing behavior of network affects the timing properties, we also consider the behavior of the underlying network on scheduling messages while modeling the communication substrate. When the number of devices increases, we may encounter state space explosion problem during model checking. To tackle the problem, we propose a reduction technique while preserving the timing properties of the communication patterns, and we prove the correctness of our technique. We implement our reduction technique in Java and build a tool that automatically reduces the state space generated by Afra. We illustrate the applicability of our reduction technique through two case studies on two clinical scenarios made of several instances of patterns. Our experimental result shows that our reduction technique can minimize the number of states significantly and make analysis of larger systems possible.
	%\fixme{Isn't the number very much dependent on the example?}
	The contributions of the paper can be summarized as: 
	\begin{itemize}
		%\REM{\item Modeling the communication patterns \REM{in terms of} using Timed Rebeca;}
		\item Modeling the communication patterns using Timed Rebeca and providing templates for building Timed Rebeca models of composite medical systems that are connected based on the communication patterns; %can be specified in terms of the models of patterns;
		%\item {\fixme{\FG{Considering the network in the model of communication substrate}}};
		\item Proposing a novel technique for state space reduction in model checking of Timed Rebeca models;
		%\REM{Extending the reduction technique of Timed Rebeca to reduce the state-space;}
		\item Modeling and analyzing two real-world case-studies.
	\end{itemize}
	
	This paper extends an earlier conference publication \cite{COORD20} by adding more explanation on the theory and foundation of our reduction technique.
	We provide a visualization of the state-space to show the reduction in a clearer way. We also provide guidelines and templates for modeling composite systems. The experiments on considering different communication substrate models due to different networks and the first case-study are also new materials. In our communication substrate models, we model the effect of networks by introducing different timing delays or priority on transmitting the messages of patterns.

	The novelty of our modeling approach %of a composite medical system 
	is that only the behavior of devices/apps need to be modeled. Thanks to the interface components, the behavior of devices/apps are separated from the ones monitoring the local timing properties. The model of patterns used in the system is reused with no modification and the proposed template for the communication substrate should be only adjusted to handle the messages of involved patterns. Our reduction technique takes advantage of static properties of patterns to merge those states satisfying the same local timing properties of communication patterns.
	
	Although the approach in this paper is motivated by needs of interoperable medical systems, the communication patterns and architectural assumptions that underlie the approach are application-independent. Thus, approach can also be used in other application domains in which systems are built from middleware-integrated components as long as the communications used in this paper are applied for specifying intercomponent communication.

	\section{Communication Patterns}\label{SS:patterns}
	
	In this paper, we model communication patterns using Rebeca, here we provide an outline of patterns (based on the content of \cite{7318707}).
	Devices and apps involved in a communication pattern are known as components that communicate with each other via a \emph{communication substrate}, e.g., networking system calls or a middleware.
	Each pattern is composed of a set of roles accomplished by components. A component may participate in several patterns with different roles simultaneously. Patterns are parameterized by a set of local timing properties that their violation can lead to a failure. In addition, each pattern has a point-to-point timing requirement that should be guaranteed by communication substrate.
	There are five communication patterns:
	\begin{itemize}
		\item\textbf{Publisher-Subscriber:}
		a publisher role broadcasts data about a topic and every devices/apps that need it can subscribe to data. Publisher does not wait for any acknowledgement or response from subscribers, \JH{so communication is asynchronous and one-way}.
		\item\textbf{Requester-Responder:}
		a requester role requests data from a specific responder and waits for data from the responder.
		\item\textbf{Initiator-Executor:}
		an initiator role requests a specific executor to perform an action and waits for action completion or its failure.
		\item \textbf{Sender-Receiver:}
		a sender role sends data to a specific receiver and waits until either data is accepted or rejected.
		%     \item\textbf{Orchestration:}
		%     an orchestrator role controls other components through different % communication patterns described above to achieve a specific goal.
	\end{itemize}
	
	%\FG{We explain Publisher-Subscriber,  Requester-Responder, and Initiator-Executor communication patterns as we use them in our case studies. Sender-Receiver pattern resembles Initiator-Executor pattern (with the same parameters and timing requirement).} 
	
	\subsection{Publisher-Subscriber\label{subsec::ppparam}}
	In this pattern, the component with the publisher role sends a $\it publish$ message to those components that have subscribed previously.
	\JH{Even when there is only a single subscriber component, choosing the pattern may be appropriate in situations where one wishes to have one-way asynchronous communication between the sender and receiver. In the interoperable medical device domain, this pattern would commonly be used in situations where a bedside monitor such as a pulse oximeter is sending data such as pulse rate information ($PR$) or blood oxygenation information ($SpO_2$) to some type of remote display (like a monitor that aggregates many types of health-related information for the patient) and/or applications that watch for trends in data to generate alarms for care-givers or to trigger some type of automated change in the patient's treatment. In these situations, there is a one-way flow of information from the monitoring device to one or more consumers.}
	
	This pattern is parameterized with the following local timing properties:
	\begin{itemize}
		\item MinimumSeparation ($N_{pub}$): if the interval between two consecutive $\it publish$ messages from the publisher is less than $N_{pub}$, then the second one is dropped by announcing a \textit{fast Publication} failure.
		\item MaximumLatency ($L_{pub}$): if the communication substrate fails to accept $\it publish$ message within $L_{pub}$ time units, it informs the publisher of \textit{timeout}.
		\item MinimumRemainingLifeTime ($R_{pub}$): if the data arrives at the subscriber late, i.e., after $R_{pub}$ time units since publication, the subscriber is notified by a  \textit{stale data} failure.
		\item MinimumSeparation ($N_{sub}$): if the interval between arrival of two consecutive messages at the subscriber is less than $N_{sub}$,  then the second one is dropped.
		\item MaximumSeparation ($X_{sub}$): if the interval between arrival of two consecutive messages at the subscriber is greater than $X_{sub}$ then the subscriber is notified by a \textit{slow publication} failure.
		\item MaximumLatency ($L_{sub}$): if the subscriber fails to consume a message within $L_{sub}$ time units, then it is notified by a \textit{slow consumption} failure.
		\item MinimumRemainingLifeTime ($R_{sub}$): if the remaining life time of the $\it publish$ message is less than $R_{sub}$, then the subscriber is notified by a \textit{stale data} failure.
	\end{itemize}
	
	\JH{The timing properties are chosen to (a) enable both the producer and consumer to characterize their local timing behavior or requirements and (b) enable reasoning about the producer/consumer time behavior compatibility and important ``end-to-end" timing properties when a producer and consumer are composed. For example, the local property $N_{pub}$ allows the publisher to specify the minimum separation time between the messages that it will publish. From this value, one can derive the maximum rate at which messages will be sent. This provides a basis for potential consumer components to determine if their processing capabilities are sufficient to handle messages coming at that rate. Publisher compliance to $N_{pub}$ can be checked at run-time within the communication infrastructure of the producer, e.g. before outgoing messages are handed off to the communication substrate, by storing the time that the previous message was sent.   Similarly, $N_{sub}$ and $X_{sub}$ are local properties for the subscriber. These allow the subscriber to state its assumptions/needs about the timing of incoming data. Figure~\ref{Fig:PubSub} gives further intuition about the purpose and relationship between the parameters.}

	\JH{While the above parameters can be seen as part of component \emph{interface specifications} on both the Publisher and Subscriber components, when reasoning about end-to-end properties, the following attribute reflects a property of the networking resource upon which inter-component communication is deployed.
		
		\begin{itemize}
			\item MaximumLatencyOfCommSubstrate ($L_{m}$): the maximum latency of communication of messages between the Publisher and Subscriber across the communication substrate. 
		\end{itemize}
	}
	
	Each communication pattern owns a non-local point-to-point \textit{timing requirement} that considers aggregate latencies across the path of the communication -- including delays introduced by application components, interfaces, and the communication substrate ($L_{m}$). 
	% that should be guaranteed by the communication substrate. 
	In this pattern the requirement is ``the data to be delivered with lifetime of at least $R_{sub}$, communication substrate should ensure maximum message delivery latency [across the substrate] $L_m$ does not exceed $R_{pub}-R_{sub}-L_{pub}$''
	(inequality \ref{eq:pub-sub}).
	
	\begin{equation}
	\label{eq:pub-sub}
	R_{pub}-R_{sub}-L_{pub} \geq L_{m}
	\end{equation}
	
	\JH{Regarding the intuition of this inequality, consider that the publisher will send a piece of data with a parameter $R_{pub}$ indicating how long that data will be fresh/valid. As the message is communicated, latencies will accumulate in the PublisherInterface (maximum value is $L_{pub}$) and communication substrate (maximum value is $L_{m}$). When the message arrives at the SubscriberInterface, it’s remaining freshness would be $R_{pub} - L_{pub} - L_{m}$.   That remaining freshness should be as least as large as $R_{sub}$ — the time needed by the subscriber to do interesting application work with the value, i.e., to achieve the goals of the communication, the following inequality should hold $R_{pub} - L_{pub} - L_{m} \geq R_{sub}$.
		The intuition is that $R_{pub}$ is a application property of the publisher — in essence, a ``guarantee'' of freshness to consumers based on the nature of the data, and $R_{sub}$ is ``requirement'' of the consumer (it needs data at least that fresh to do its application work). Given that $L_{pub}$ is a fixed latency in the software that interacts with the network, the network needs to guarantee that $L_{m}$ is low enough to make the inequality above hold. Using algebra to reorient the constraint so that it can be more clearly represented as a latency constraint on the communication substrate ($L_m$) yields Inequality~\ref{eq:pub-sub}.} 
	
	\JH{For an example of how these parameters might be used in a medical application, assume a pulse oximeter device that publishes pulse rate data of the patient. A monitoring application might subscribe to the physiological readings from the pulse oximeter and other devices to support a ``dashboard" that provides a composite view device readings and generates alerts for care-givers based on a collection of physiological parameters. In such a system, the Publisher-Subscriber pattern can be used to communicate information from the pulse oximeter (publisher) to the monitoring application (subscriber). In this description, there is only one subscriber (the monitoring application), but using the Publisher-Subscriber pattern is still appropriate because it allows other subscribers (e.g., a separate alarm application, or a data logging application) to be easily added. Even when there is a single subscriber, the pattern selection emphasizes that the communication is one-way. For publisher local properties, the pulse oximeter can use $N_{pub}$ to indicate the maximum rate at which it will publish blood oxygenation information (Sp$O_2$) and/or pulse rate information. In medical devices in general, this rate would typically be associated with the interval at which meaningful changes can be reflected in the reported physiological parameters. The device designer would use the $L_{pub}$ parameter to specify the maximum length of the delay associated with putting a published value out on the communication substrate that would be acceptable for safe and correct use of the device. On the subscriber side, $N_{sub}$ allows the monitoring application to specify an upper bound on the rate of incoming messages. The value chosen may be derived in part from the execution time needed to compute new information and format resulting data for the display. Intuitively, the $X_{sub}$ allows the monitoring application to indicate how frequently it needs pulse oximetry data to maintain an ``up to date'' display.
		
		The other properties can be used to characterize end-to-end (non-local) timing concerns. To ensure that care-givers receive timely dashboard information and alerts, safety requirements should specify that information is (a) communicated from the pulse oximeter device to the monitoring application with a medically appropriate bound on the latency, and (b) the received physiological parameter is currently an accurate reflection of the patient's physiological state  (i.e., the parameter is “fresh” enough to support the medical intended use). Such requirements would build on the type of non-local timing requirement specified above.
	}
	
	% For example assume a pulse oximeter device which publishes pulse rate data of the patient
	% A patient monitor application can subscribe to this data to get the patient's pulse rate. 
	% In other words, the application communicates with the device using the Publisher-Subscriber pattern.

	\subsection{Requester-Responder}
	In this pattern, the component with the role requester, sends a $\it request$ message to the component with the role responder. The responder should reply within a time limit as specified by its local timing properties.  
	\JH{In the interoperable medical device domain, this pattern would commonly be used in situations where an
		application needs to ``pull'' information from a medical device (e.g., retrieving the current blood pressure reading from a blood pressure device, retrieving the infusion settings from an infusion pump) or fetching patient data from medical record database.}
	
	This pattern is parameterized with the following local timing properties:
	\begin{itemize}
		\item MinimumSeparation ($N_{req}$): if interval between two consecutive $\it request$ messages is less than $N_{req}$, then the second one is dropped with a \textit{fast Request} failure.
		\item MaximumLatency ($L_{req}$): if the $\it response$ message does not arrive within $L_{req}$ time units, then the request is ended by a \textit{timeout} failure.
		\item MinimumRemainingLifeTime ($R_{req}$): if the $\it response$ message arrives at the requester with a remaining lifetime less than $R_{req}$, then the requester is notified by a \textit{stale data} failure.
		\item MinimumSeparation ($N_{res}$): if the duration between the arrival of two consecutive $\it request$ messages is less than $N_{res}$, then the request is dropped while announcing an  \textit{excess load} failure.
		\item MaximumLatency ($L_{res}$): if the $\it response$ message is not provided within the $L_{res}$ time units, the request is ended by a \textit{timeout} failure.
		\item MinimumRemainingLifeTime ($R_{res}$): if the $\it request$ message with the promised minimum remaining lifetime cannot be responded by the responder, then request is ended by a \textit{data unavailable} failure.
	\end{itemize}
	
	\JH{Compared to the Publisher-Subscriber, several of the timing specification parameters are similar, while others are reoriented to focus on the completion of the end-to-end two-phase "send the request out, get a response back" as opposed to the one-phase goal of the Publisher-Subscriber "send the message out". For example, the minimum separation parameters for both the Requester $N_{req}$ and Responder $N_{res}$ are analogous to the $N_{pub}$ and $N_{sub}$ parameters of the Publish-Subscriber pattern. The MinimumRemainingLifetime concept is extended to include a check not only on the arrival of the request at the responder ($R_{req}$) at the first phase of the communication, but also the a check on the communication from the Responder back to the Requester (at the end of the second phase of the communication).  
	}
	
	Reasoning about the end-to-end two-phase objective of this pattern now needs to consider communication substrate latencies for both the request message $L_{m}$ and the response message $L'_{m}$.  
	The point-to-point \textit{timing requirement} defined for this pattern concerns the delivery of response with lifetime of at least $R_{req}$. So the communication substrate should ensure that ``the sum of [its] maximum latencies to deliver the request to the responder ($L_m$) and the resulting response to the requester ($L'_m$) does not exceed $L_{req}+R_{req}-L_{res}-R_{res}$''
	(inequality \ref{eq:req-res}).
	\\
	\begin{equation}
	\label{eq:req-res}
	L_{req}+R_{req}-L_{res}-R_{res} \geq L_{m}+L'_{m}
	\end{equation}
	
	\JH{
		For an example of how this pattern might be used in a medical application, consider a medical application that requires a blood pressure reading. The application would send a request message to a blood pressure device (with maximum communication substrate latency $L_{m}$), the blood pressure device would either return the most recent reading or acquire a new reading (with latency of $L_{res}$ to obtain the value within the device), and then the device would send a response message to the requester (with maximum communication substrate latency $L_{m}$).  $L_{req}$ expresses the application's requirement on the overall latency of the interaction. The lifetime parameters can be used in a manner similar to that of the Publisher-Subscriber pattern.}
	
	% Replace the following with the explanation above: For example assume a patient monitor application that communicates with a blood pressure (BP) monitor using the Requester-Responder pattern. The application requests blood pressure measurement from the BP which periodically measures the blood pressure of the patient.

	\subsection{Initiator-Executor}
		In this pattern, the component with the initiator role, requests   a specific component with the executor role to execute an action. The executor should provide appropriate acknowledgment message (action succeeded, action failed or action unavailable) within a time limit as specified by its local timing properties.  \JH{In interoperable medical applications, this pattern would typically be used by an application to instruct an actuation device to perform some action. For example, an infusion control application might use the pattern to start or stop the infusion process. A computer-assisted surgery application might use the pattern to instruct the movement of computer-controlled surgical instruments.}
		
		This pattern is parameterized with the following local timing properties:
		\begin{itemize}
			\item MinimumSeparation ($N_{ini}$): if interval between two consecutive $\it initiate$ messages is less than $N_{ini}$, then the second one is dropped with a \textit{fast init} failure.
			\item MaximumLatency ($L_{ini}$): if the $\it acknowledgment$ message does not arrive within $L_{ini}$ time units, then the request is ended by a \textit{timeout} failure.
			\item MinimumSeparation ($N_{exe}$): if the duration between the arrival of two consecutive $\it initiate$ message is less than $N_{exe}$, then the request is dropped while announcing an  \textit{excess load} failure.
			\item MaximumLatency ($L_{exe}$): once the initiating message arrives at the Executor, if the $\it acknowledgment$ message is not provided within the $L_{exe}$ time units, the request is ended by a \textit{timeout} failure.
		\end{itemize}
		
		\JH{$N_{ini}$ can be seen as a guarantee in interface specification on the Initiator to not send messages faster than a certain rate.  $L_{ini}$ is a requirement that the Initiator has on the overall latency of the action. Failure of the system to satisfy this property might lead the initiating component to raise an alarm or take some other corrective action necessary for safety.
			$N_{exe}$ can be understood as a requirement that the Executor has related to its ability to handle action requests.  $L_{exe}$ can be seen as a guarantee that the Executor provides to either perform the action or generate a time out message within a certain time bound.}
		
		The point-to-point \textit{timing requirement} defined for this pattern concerns the delivery of data within maximum latencies.  \JH{The overall latency of the actions of sending the execution command $L_m$, the executor carrying out the action $L_{exe}$, and sending the acknowledge $L'_m$ should not exceed the requirement on the overall latency specified by the Initiator $L_{ini}$.
			Some algebra on this relationship to focus on the requirements of the communication substrate yields the following inequality
			% So the communication substrate should ensure that ``the sum of maximum latencies to deliver the initiation to the executor ($L_m$) and the resulting acknowledgment to the initiator ($L'_m$) does not exceed $L_{ini}-L_{exe}$''
			(inequality \ref{eq:ini-exe}).}
		\\
		\begin{equation}
		\label{eq:ini-exe}
		L_{ini}-L_{exe} \geq L_{m}+L'_{m}
		\end{equation}
		
		For example, in the X-Ray/Ventilator synchronization in Section~\ref{sec:x-ray-app}, a coordinating application needs to send commands to both the X-Ray and Ventilator. The Initiator-Executor pattern can be used to control both of these devices with the minimum separation constraints as used in previous patterns. The parameter $L_{ini}$ parameter would be used to specify the requirement on the maximum latency of each interaction.
		% For example assume a patient that hooked to a ventilator and needs to be radiographed. A coordination app should initiate a start/stop action on the ventilator.

	\subsection{Sender-Receiver}
		In this pattern, the component with the sender role, sends data to a specific component with the receiver role. 
		The receiver should reply with appropriate acknowledgment message (data accepted or data rejected) within a time limit as specified by its local timing properties. 
		\JH{This pattern is structurally and semantically very similar to the Initiator-Executor pattern. It is only presented as a separate pattern to distinguish the fact that the receiving component only accepts data and, e.g., stores it rather than performing an action that may impact the external environment.}
		
		This pattern is parameterized with the following local timing properties:
		\begin{itemize}
			\item MinimumSeparation ($N_{sen}$): if interval between two consecutive $\it send$ messages is less than $N_{sen}$, then the second one is dropped with a \textit{fast send} failure.
			\item MaximumLatency ($L_{sen}$): if the $\it acknowledgment$ message does not arrive within $L_{sen}$ time units, then the data sent is ended by a \textit{timeout} failure.
			\item MinimumSeparation ($N_{rec}$): if the duration between the arrival of two consecutive $\it send$ messages is less than $N_{rec}$, then the data is dropped while announcing an  \textit{excess load} failure.
			\item MaximumLatency ($L_{rec}$): if the $\it acknowledgment$ message is not provided within the $L_{rec}$ time units, the data sent is ended by a \textit{timeout} failure.
		\end{itemize}
		
		The point-to-point \textit{timing requirement} defined for this pattern concerns the delivery of data within maximum latencies. So the communication substrate should ensure that ``the sum of maximum latencies to deliver the sent data to the reciever ($L_m$) and the resulting acknowledgment to the sender ($L'_m$) does not exceed $L_{sen}-L_{rec}$''
		(inequality \ref{eq:sen-rec}).
		\\
		\begin{equation}
		\label{eq:sen-rec}
		L_{sen}-L_{rec} \geq L_{m}+L'_{m}
		\end{equation}
		
		\JH{In interoperable medical applications, this pattern would typically be used to change the settings on a device or to update a record in some electronic medical record. For example assume a BP monitor that measures blood pressure every 3 minutes periodically. The monitoring application could use the pattern to change the settings on the device to an interval of 1 minute.}
	
	\section{Timed Rebeca and Actor Model}
	Actor model \cite{agha1985actors,hewitt1977viewing} is a concurrent model based on computational objects, called actors, that communicate asynchronously with each other. Actors are encapsulated modules with no shared variables. Each actor has a unique address and mailbox. Messages sent to an actor are stored in its mailbox. Each actor is defined through a set of message handlers to specify the
	actor behavior upon processing of each message.
	
	Rebeca \cite{sirjani2004modeling,Sirjani06} is an actor model language with a Java-like syntax which aims to bridge the gap
	between formal verification techniques and the real-world software engineering of concurrent and distributed applications. Rebeca is supported by
	a robust model checking tool, named Afra\footnote{\url{http://www.rebeca-lang.org/alltools/Afra}}. Timed Rebeca is an extension of Rebeca for modeling and verification of concurrent and distributed systems with timing constraints. As all timing properties in communication patterns are based on time, we use Timed Rebeca for modeling and formal analysis of patterns by Afra.
	Hereafter, we use Rebeca as short for Timed Rebeca in the paper.
	
	\begin{figure}[hbtp]
		\begin{center}
			\fbox{\parbox{\columnwidth-5mm}
				{\footnotesize
					\begin{align*}
					\mathrm{Model} &\Coloneqq~\langle{Class}\rangle^+~\mathrm{Main}~\\
					\mathrm{Main} &\Coloneqq ~\mathsf{main}~ \{\mathrm{InstanceDcl}^*\} \\
					\mathrm{InstanceDcl} &\Coloneqq ~\mathrm{C}~ \mathrm{r}~(\mathrm{\langle r\rangle}^* )~\colon(\mathrm{\langle c\rangle}^*)\\
					\mathrm{Class} &\Coloneqq~ \mathsf{reactiveclass} ~\mathrm{C}~ \{\mathsf{KnownRebecs}~ \mathsf{Vars}~ \mathsf{MsgSrv}^* \}\\
					\mathrm{KnowRebecs} &\Coloneqq~ \mathsf{knownrebecs}~\{\mathrm{VarDcl} \}\\
					\mathrm{Vars} &\Coloneqq \mathsf{statevars}~ \{\mathrm{VarDcl} \}\\
					\mathrm{VarDcl} &\Coloneqq \langle\mathrm{T}~ \mathrm{v}\rangle^* ;\\
					\mathrm{MesgSrv} &\Coloneqq \mathsf{msgsrv} ~\mathrm{m}~(\mathrm{VarDcl})~\{\mathrm{Stmt}^* \}\\
					\mathrm{Stmt}
					&\Coloneqq~
					\mathrm{v}=e;~|~\mathrm{Call};~|~\mathsf{if}(e)~\mathrm{MSt}~[\mathsf{else}~\mathrm{MSt}]~|~ \mathsf{delay}(\mathrm{t});\\
					\mathrm{Call}
					&\Coloneqq~\mathrm{r.m(\langle e\rangle ^*)}[\mathsf{deadline} ~e][\mathsf{after}~ e]\\
					\mathrm{MSt}
					&\Coloneqq~\{\mathrm{Stmt}^*\}~|~ \mathrm{Stmt}% \\
					%\mathrm{Expr}
					%&\Coloneqq~c~|~v~|~\mathrm{Expr}~ \mathsf{op}~ \mathrm{Expr},~\mathsf{op}\in\{+,-,*,\wedge,\vee,<,\le,>,\ge\} ~|~ %(\mathrm{Expr}) ~|~\mathsf{!}(\mathrm{Expr)}
					\end{align*}
			}}
		\end{center}
		\vspace{-2mm} \caption{Abstract syntax of Timed Rebeca. Angle brackets $\langle~\rangle$ denotes meta parenthesis, superscripts $+$ and $*$ respectively are used for repetition of one or more and repetition of zero or more times. Combination of  $\langle~\rangle$ with repetition is used for comma separated list. Brackets $[~]$ are used for optional syntax. Identifiers $C$, $T$, $m$, $v$, $c$, $e$, and $r$ respectively denote class, type, method name, variable, constant, expressions, and rebec name, respectively.
			\label{Fig::TimedRebecaGrammar}}
	\end{figure}
	
	The syntax of Timed Rebeca \cite{TimedRebeca,SirjaniK16} is given in Figure  \ref{Fig::TimedRebecaGrammar}. Each Rebeca model contains \textit{reactive classes} definition and \textit{main} part. Main part contains instances of reactive classes. These instances are actors that are called rebecs. Reactive classes have three parts: \textit{known rebecs}, \textit{state variables} and \textit{message servers}.
	Each rebec can communicate with its known rebecs or itself. Local state of a rebec is indicated by its state variables and received messages which are in the rebec's mailbox. Rebecs are reactive, there is no explicit receive and mailbox manipulation. 
	\EK{Messages trigger the execution of the statements of message
		servers when they are taken from the
		message mailbox. An actor can change its state variables through assignment statements, make decisions through conditional statements, communicates with other actors by sending messages, and performs periodic behavior by sending messages to itself. A message server may have a \emph{nondeterministic assignment} statement which is used to model the nondeterminism in the behavior of a message server. The timing features are \textit{computation time}, \textit{message delivery time} and \textit{message expiration}. Computation time is shown by \textit{delay} statement. Message delivery and expiration times are expressed by associating \textit{after} and \textit{deadline} values with message sending statements}.
	%Messages trigger the execution of the message servers when they are taken from the message mailbox. The timing features are \textit{computation time}, \textit{message delivery time} and \textit{message expiration}. These three primitives are supported by the statements \textit{delay}, \textit{after} and \textit{deadline}.

	\begin{example}
		A simple request-response system is specified in Timed Rebeca given in Figure \ref{Fig:ُTimedRebecaModelingExample}. This model has two rebecs: $\it req$ is an instances of class $\it Requester$ while $\it res$ is an instance of $\it Responder$. The size of the rebec mailboxes is specified by $(5)$ after the name of classes. These two rebecs are passed as the known rebecs of each other in lines $27-28$ by instantiating. Each class has a message server with the same name as the class name and it acts similar to class constructors in object-oriented languages. Rebec $\it req$ initially sends a message $\it request$ to itself upon executing its constructor. The global time is initially $0$. Rebec $\it req$ takes the message $\it request$ from its mailbox to handle. By executing the statement ${\it delay}(3)$ it is blocked until time $3$. %As no rebec can have progress, the global time advances to $3$. 
		The rebec $\it req$ resumes its execution at time $3$ by sending a ``$\it request$" message to the rebec $\it res$. This message is delivered to the rebec $\it res$ after a delay of $8$, i.e., $11$. At time $11$, rebec $\it res$ takes the message ``$\it request$" from its mailbox. Upon executing its message server, it sends a message ``$\it response$" to $\it req$ which will be delivered at time $16$. Rebec $\it req$ takes the message ``$\it response$" from its mailbox at time $16$ and sends a message ``$\it request$" to itself.  %This message is immediately inserted into the bag of $\it res$ by setting its arrival time to $11$. Now the bag of $\it req$ is empty but $\it res$ can not handle its message by time $11$. So, the time progress to $11$. The state-space of model has been illustrated in Figure \ref{}  
		
	\end{example}
	\begin{figure}[htbp]
\begin{lstlisting}[style=customjava,label=ُTimedRebecaModelingExampl,multicols=2]
reactiveclass Requester(5){
    knownrebecs{
      Responder res;
    }
    Requester(){
       self.request();
    }
    msgsrv request(){
       delay(3);
       res.request() after(8);
    }
    msgsrv response(){
       self.request();
    }
}

reactiveclass Responder(5){
    knownrebecs{
       Requester req;
    }
    msgsrv request(){
       req.response() after(5);
    }
}

main {
   Requester req(res):();
   Responder res(req):();
}
\end{lstlisting}
		\caption{A simple request-response system in Timed Rebca}
		\label{Fig:ُTimedRebecaModelingExample}
	\end{figure}

	\subsection{Standard Semantics of Timed Rebeca Model}
	\EK{Formal semantics of Timed Rebeca models are presented as labeled transition systems (LTS), i.e., a basic model to define the semantics of reactive systems \cite{baier}. A LTS is defined by the quadruple $\langle S,\rightarrow,L,s_{0} \rangle$ where $S$ is
		a set of states, $L$ a set of labels, $\rightarrow\,\subseteq S\times L\times S$ a set of
		transitions, and $s_{0}$ the initial state. Let
		$s\overto{\alpha}t$ denote $(s,\alpha,t)\in\rightarrow$. Timed Transition System (TTS),  a basic computational model of realtime systems, generalizes the basic computation model of transition systems by associating an interval with each transition to indicate how long a transition takes \cite{DBLP:conf/rex/HenzingerMP91}. In a TTS, transitions are partitioned into two classes: instantaneous transitions (in which time does not progress), and time ticks when the global clock is incremented. These time ticks happen when all participants ``agree'' for time elapse. The standard semantics of Timed Rebeca is defined in terms of TTS as described in \cite{DBLP:journals/scp/KhamespanahKS18}. 
		
		In the following, the brief description of this semantics is presented based on \cite{DBLP:journals/scp/KhamespanahKS18}. Assume that $\mathit{AID}$ is the set of the identifiers of all of the rebecs, $\mathit{MName}$ is the set of the names of all of the message servers, $\mathit{Var}$ is the set of the all of the identifiers of variables, and $\mathit{Val}$ is all of the possible values of variables. Also, $\mathit{Msg} = \mathit{AID} \times \mathit{MName} \times (\mathit{Var} \rightarrow \mathit{Val}) \times \mathbb{N} \times \mathbb{N}$ is the type of messages which are passed among actors. In a message $(i,m,r,a,d) \in \mathit{Msg}$, $i$ is the identifier of the sender of this message, $m$ is the name of its corresponding method, $r$ is a function mapping argument names to their values, $a$ is its arrival time, and $d$ is its deadline. Also assume that the set $\powerset{\mathit{A}}$ is the power set and $\powermultiset{\mathit{A}}$ is the power multiset of its given set $A$.
		
		\begin{defi}
			For a given Timed Rebeca model $\mathcal{M}$, 
			%and property specification $\mathcal{P}_\mathcal{M}$ 
			$\mathit{TTS_{\mathcal{M}}}=(S, \rightarrow, Act, s_0)$ is its standard semantics such that:
			\begin{itemize}
				\item The global state of a Timed Rebeca model is represented by a function $s : \mathit{AID} \rightarrow (\mathit{Var} \rightarrow \mathit{Val}) \times \powermultiset{\mathit{Msg}} \times \mathbb{N} \times \mathbb{N} \times \mathbb{N} \cup \{\epsilon\}$, which maps an actor's identifier to the local state of the actor. The local state of an actor is defined by a tuple like $(v, q, \sigma, t, r)$, where $v : \mathit{Var} \rightarrow \mathit{Val}$ gives the values of the state variables of the actor, $q : \powermultiset{\mathit{Msg}}$ is the message bag of the actor, $\sigma : \mathbb{N}$ is the program counter, $t$ is the actor local time, and $r$ is the time when the actor resumes executing remaining statements. The value of $\epsilon$ for the resuming time shows that this actor is not executing a message server. %Note that actors communicate via message passing and put their incoming messages into message bags.
				\item In the initial state of the model, for all of the actors, the values of state variables and content of the actor's message bag is set based on the statements of its constructor method, and the program counter is set to zero. The local times of the actors are set to zero and their resuming times are set to $\epsilon$.
				\item The set of actions is defined as $Act = \mathit{MName} \cup \mathbb{N} \cup \{\tau\}.$
				\item The transition relation $\rightarrow \,\subseteq S \times Act \times S$ defines the transitions between states that occur as the results of actors' activities including: taking a message from the mailbox, executing a statement, and progress in time. The latter is only enabled when the others are disabled for all of the actors. This rule performs the minimum required progress of time to enable one of the other rules. %Note that we associated a rule name with $\tau$ transitions to relate $\tau$ transitions to their corresponding rules.
			\end{itemize}
		\end{defi}
		More details and SOS rules which define these transitions are presented in \cite{DBLP:journals/scp/KhamespanahKS18}.
	}
	
	\begin{example}
		We explain the state-space shown in Figure \ref{fig::TTS} derived partially for the Rebeca model given in Figure \ref{Fig:ُTimedRebecaModelingExample}. The global state is defined by the local states of rebecs and global time. 
		\EK{Note that this presentation has some minor difference in comparison with the structure of the global state in the presented semantics. As local states of all the rebecs in TTS has the same time, in Figure \ref{Fig:ُTimedRebecaModelingExample} one value for \emph{now} is shown as the global time of the system. In addition, the values of the state variables and resuming times are omitted to make the figure simpler.} 
		%In this figure, the global state is shown by the simplified version instead of having the same values for \emph{now} of all rebecs, ha
		In the initial state, called $s_1$, only the rebec $\it req$ has a message ``$\it request$" in its bag. By taking this message, we have a transition of type event to the state $s_2$ while the $\it pc$ of rebec is set to $1$ indicating the first statement of the message server ``$\it request$" should be executed. Upon executing the delay statement, the rebec is suspended for 3 units of time. As no rebec can have progress, the global time advances to $3$ and there is a time transition to the state $s_3$. Now, rebec $\it req$ resumes its execution by executing the send statements. This execution makes a state transition to the state $s_4$ by inserting a message ``$\it request$" into the mailbox of $\it res$ by setting its arrival time to $11$. %Now the bag of $\it req$ is empty but $\it res$ can not handle its message by time $11$. So, the time progress to $11$. The state-space of model has been illustrated in Figure \ref{}  
	\end{example}

	\begin{figure}
		\centering
		\begin{subfigure}[b]{0.3\textwidth}
			%\subfigure[TTS]{
			%\label{fig::TTS}
			\centering
			\small{
				\tikzstyle{line} = [draw, -stealth, thick]
\begin{tikzpicture}[scale=0.6, transform shape]
%\begin{latin}
\node[rectangle](s1)
{%\begin{latin}
	\begin{tabular}{|l|l|l|} 
	\hline
	\multicolumn{3}{|c|}{$s_1$}       \\ 
	\hline
	\multirow{2}{*}{\rotatebox[origin=c]{90}{req}} & bag  & [(request(),0,$\infty$)]  \\ 
	\cline{2-3}
	& pc &   \\ 
	\hline
	\multirow{2}{*}{\rotatebox[origin=c]{90}{res}} & bag  & []  \\ 
	\cline{2-3}
	& pc &   \\
	\hline
	\rotatebox[origin=c]{90}{now} & \multicolumn{2}{|c|}{$0$}\\
	\hline
	\end{tabular}
	%\end{latin}
	};
%\\\\\\\\\\\\\\\\\\\\\\\\\\\\\\\\\\\\\\\\\
\node[rectangle, below of=s1, yshift=-8em](s2)
{%\begin{latin}
	\begin{tabular}{|l|l|l|} 
	\hline
	\multicolumn{3}{|c|}{$s_2$}       \\ 
	\hline
	\multirow{2}{*}{\rotatebox[origin=c]{90}{req}} & bag  & []  \\ 
	\cline{2-3}
	& pc &  request:1 \\ 
	\hline
	\multirow{2}{*}{\rotatebox[origin=c]{90}{res}} & bag  & []  \\ 
	\cline{2-3}
	& pc &   \\
	\hline
	\rotatebox[origin=c]{90}{now} & \multicolumn{2}{|c|}{$0$}\\
	\hline
	\end{tabular}
	%\end{latin}
	};
%\\\\\\\\\\\\\\\\\\\\\\\\\\\\\\\\\\\\\\\\\
\node[rectangle, below of=s2, yshift=-8em](s3)
{%\begin{latin}
	\begin{tabular}{|l|l|l|} 
	\hline
	\multicolumn{3}{|c|}{$s_3$}       \\ 
	\hline
	\multirow{2}{*}{\rotatebox[origin=c]{90}{req}} & bag  & []  \\ 
	\cline{2-3}
	& pc &  request:2 \\ 
	\hline
	\multirow{2}{*}{\rotatebox[origin=c]{90}{res}} & bag  & []  \\ 
	\cline{2-3}
	& pc &   \\
	\hline
	\rotatebox[origin=c]{90}{now} & \multicolumn{2}{|c|}{$3$}\\
	\hline
	\end{tabular}
	%\end{latin}
	};
%\\\\\\\\\\\\\\\\\\\\\\\\\\\\\\\\\\\\\\\\\
\node[rectangle, below of=s3, yshift=-8em](s4)
{%\begin{latin}
	\begin{tabular}{|l|l|l|} 
	\hline
	\multicolumn{3}{|c|}{$s_4$}       \\ 
	\hline
	\multirow{2}{*}{\rotatebox[origin=c]{90}{req}} & bag  & []  \\ 
	\cline{2-3}
	& pc &  \\ 
	\hline
	\multirow{2}{*}{\rotatebox[origin=c]{90}{res}} & bag  & [(request(),11,$\infty$)]  \\ 
	\cline{2-3}
	& pc &   \\
	\hline
	\rotatebox[origin=c]{90}{now} & \multicolumn{2}{|c|}{$3$}\\
	\hline
	\end{tabular}
	%\end{latin}
	};

\node (s0) at (0,-13.5) {$\ldots$};
\path [line] (s1) -- node[xshift=-4em] {(request(),0,$\infty$)} (s2);
\path [line] (s2) -- node[xshift=-2em] {time+=3} (s3);
\path [line] (s3) -- node[xshift=-2em] {$\tau(req)$} (s4);
\path [line] (s4) -- node[xshift=-4em] {}(s0);
%\end{latin}
\end{tikzpicture}
			}
			\caption{TTS\label{fig::TTS}}
			%}
		\end{subfigure}
		%\qquad
		\begin{subfigure}[b]{0.3\textwidth}
			%\subfigure[FTTS]{
			%\label{fig::FTTS}
			\centering
			\small{
				\tikzstyle{line} = [draw, -stealth, thick]
\begin{tikzpicture}[scale=0.6, transform shape]
%\begin{latin}
\node[rectangle](s1)
{%\begin{latin}
	\begin{tabular}{|l|l|l|} 
	\hline
	\multicolumn{3}{|c|}{$t_1$}       \\ 
	\hline
	\multirow{2}{*}{\rotatebox[origin=c]{90}{req}} & bag  & [(request(),0,$\infty$)]  \\ 
	\cline{2-3}
	& now &  0 \\ 
	\hline
	\multirow{2}{*}{\rotatebox[origin=c]{90}{res}} & bag  & []  \\ 
	\cline{2-3}
	& now &  0 \\
	\hline
	\end{tabular}
	%\end{latin}
	};
%\\\\\\\\\\\\\\\\\\\\\\\\\\\\\\\\\\\\\\\\\
\node[rectangle, below of=s1, yshift=-8em](s2)
{%\begin{latin}
	\begin{tabular}{|l|l|l|} 
	\hline
	\multicolumn{3}{|c|}{$t_2$}       \\ 
	\hline
	\multirow{2}{*}{\rotatebox[origin=c]{90}{req}} & bag  & []  \\ 
	\cline{2-3}
	& now &  3 \\ 
	\hline
	\multirow{2}{*}{\rotatebox[origin=c]{90}{res}} & bag  & [(request(),11,$\infty$)]  \\ 
	\cline{2-3}
	& now & 0  \\
	\hline
	\end{tabular}
	%\end{latin}
	};
%\\\\\\\\\\\\\\\\\\\\\\\\\\\\\\\\\\\\\\\\\
\node[rectangle, below of=s2, yshift=-8em](s3)
{%\begin{latin}
	\begin{tabular}{|l|l|l|} 
	\hline
	\multicolumn{3}{|c|}{$t_3$}       \\ 
	\hline
	\multirow{2}{*}{\rotatebox[origin=c]{90}{req}} & bag  & [(response(),16,$\infty$)]  \\ 
	\cline{2-3}
	& now & 3 \\ 
	\hline
	\multirow{2}{*}{\rotatebox[origin=c]{90}{res}} & bag  & []  \\ 
	\cline{2-3}
	& now &  11 \\
	\hline
	\end{tabular}
	%\end{latin}
	};
%\\\\\\\\\\\\\\\\\\\\\\\\\\\\\\\\\\\\\\\\\
\node[rectangle, below of=s3, yshift=-8em](s4)
{%\begin{latin}
	\begin{tabular}{|l|l|l|} 
	\hline
	\multicolumn{3}{|c|}{$t_4$}       \\ 
	\hline
	\multirow{2}{*}{\rotatebox[origin=c]{90}{req}} & bag  & [(request(),16,$\infty$)]  \\ 
	\cline{2-3}
	& now &  16\\ 
	\hline
	\multirow{2}{*}{\rotatebox[origin=c]{90}{res}} & bag  & []  \\ 
	\cline{2-3}
	& now & 11  \\
	\hline
	\end{tabular}
	%\end{latin}
	};
%--------------
\node (s0) at (0,-13.5) {$\ldots$};
\path [line] (s1) -- node[xshift=-4em] {(request(),0,$\infty$)} (s2);
\path [line] (s2) -- node[xshift=-2em] {(request(),11,$\infty$)} (s3);
\path [line] (s3) -- node[xshift=-2em] {(response(),16,$\infty$)} (s4);
\path [line] (s4) -- node[xshift=-4em] {}(s0);
%\end{latin}
\end{tikzpicture}
				\caption{FTTS}
				\label{fig::FTTS}
			}
			%}
		\end{subfigure}
		%\qquad
		\begin{subfigure}[b]{0.3\textwidth}
			%\subfigure[FTTS]{
			%\label{fig::FTTS}
			\centering
			\small{
				\tikzstyle{line} = [draw, -stealth, thick]
\begin{tikzpicture}[scale=0.6, transform shape]
\node[rectangle](s1)
{%\begin{latin}
	\begin{tabular}{|l|l|l|} 
	\hline
	\multicolumn{3}{|c|}{$t_3'$}       \\ 
	\hline
	\multirow{2}{*}{\rotatebox[origin=c]{90}{req}} & bag  & [(response(),16,$\infty$)]  \\ 
	\cline{2-3}
	& now &  3\\ 
	\hline
	\multirow{2}{*}{\rotatebox[origin=c]{90}{res}} & bag  & []  \\ 
	\cline{2-3}
	& now & 11  \\
	\hline
	\end{tabular}
	%\end{latin}
	};
%%%%%%%%%%%%%%%%%%%%%%%%%%%%%

\node[rectangle, below of=s1, yshift=-8em](s2)
{%\begin{latin}
	\begin{tabular}{|l|l|l|} 
	\hline
	\multicolumn{3}{|c|}{$t_4'$}       \\ 
	\hline
	\multirow{2}{*}{\rotatebox[origin=c]{90}{req}} & bag  & [(request(),16,$\infty$)]  \\ 
	\cline{2-3}
	& now &  16\\ 
	\hline
	\multirow{2}{*}{\rotatebox[origin=c]{90}{res}} & bag  & []  \\ 
	\cline{2-3}
	& now & 11  \\
	\hline
	\end{tabular}
	%\end{latin}
	};
%\\\\\\\\\\\\\\\\\\\\\\\\\\\\\\\\\\\\\\\\\
\node[rectangle, below of=s2, yshift=-8em](s3)
{%\begin{latin}
	\begin{tabular}{|l|l|l|} 
	\hline
	\multicolumn{3}{|c|}{$t_5'$}       \\ 
	\hline
	\multirow{2}{*}{\rotatebox[origin=c]{90}{req}} & bag  & []  \\ 
	\cline{2-3}
	& now &  19 \\ 
	\hline
	\multirow{2}{*}{\rotatebox[origin=c]{90}{res}} & bag  & [(request(),27,$\infty$)]  \\ 
	\cline{2-3}
	& now & 11  \\
	\hline
	\end{tabular}
	%\end{latin}
	};
%\\\\\\\\\\\\\\\\\\\\\\\\\\\\\\\\\\\\\\\\\
\node[rectangle, below of=s3, yshift=-8em](s4)
{%\begin{latin}
	\begin{tabular}{|l|l|l|} 
	\hline
	\multicolumn{3}{|c|}{$t_6'$}       \\ 
	\hline
	\multirow{2}{*}{\rotatebox[origin=c]{90}{req}} & bag  & [(response(),32,$\infty$)]  \\ 
	\cline{2-3}
	& now & 19 \\ 
	\hline
	\multirow{2}{*}{\rotatebox[origin=c]{90}{res}} & bag  & []  \\ 
	\cline{2-3}
	& now &  27 \\
	\hline
	\end{tabular}
	%\end{latin}
	};
\node (s0) at (0,2) {$\ldots$};
%\\\\\\\\\\\\\\\\\\\\\\\\\\\\\\\\\\\\\\\\\

\path [line] (s0) -- node[xshift=-4em] {}(s1);
\path [line] (s1) -- node[xshift=-4em] {(request())} (s2);
\path [line] (s2) -- node[xshift=-2em] {(request())} (s3);
\path [line] (s3) -- node[xshift=-2em] {(response())} (s4);
%\path [line] (s4) -- node[xshift=-2em] {(request())} (s2);
\draw[->]  (s4) edge[out=35,in=-35] node[right]{(request())} (s2);
%\end{latin}
\end{tikzpicture}
				\caption{BFTTS}
				\label{fig::BFTTS}
			}
			%}
		\end{subfigure}
		\caption{ TTS, FTTS, and BFTTS for the Timed Rebeca model in Figure
			\ref{Fig:ُTimedRebecaModelingExample}. As rebecs in the example of Figure~\ref{Fig:ُTimedRebecaModelingExample} have no state variables, we have not shown their empty values. The $\infty$ value for the deadlines of message indicates that messages have no deadline.
		}
		\label{fig::FTTSandTTS}
	\end{figure}

	\subsection{Coarse-grained Semantics of Timed Rebeca Models}\label{subsec::FTTS}

	\EK{ Floating Time Transition System (FTTS), 
		%a variation of labeled transition systems 
		introduced in \cite{RTS}, gives a natural event-based semantics for timed actors, providing a significant amount of reduction in the size of transition systems. FTTS is a coarse-grain semantics and contains only event transitions of the standard semantics; each state transition shows the effect of handling of a message by a rebec.
		
		The semantics of Timed Rebeca in FTTS is defined in terms of a transition system. The structure of states in FTTS are the same as that of in TTS; however, the local times of actors in a state can be different. 
		%The semantics of Timed Rebeca in FTTS is defined in terms of a transition system. States in a FTTS contain the local time of each actor, in addition to values of their state variables and the bag of the received messages. However, the local times of actors in a state can be different. 
		FTTS can be used for the analysis of Timed Rebeca models as there is no shared variables, no blocking send or receive, single-threaded actors, and atomic (non-preemptive) execution of message servers which gives an isolated message server execution. As a result, the execution of a message server of an actor will not interfere with the execution of message servers of other actors. Therefore,  all the statements of a given message server can be executed (including delay statements) during a single transition. 
		
		\begin{defi}
			For a given Timed Rebeca model $\mathcal{M}$, $\mathit{FTTS_{\mathcal{M}}}=(S, \hookrightarrow, Act', s_0)$ is its floating time semantics where $S$ is the set of states, $s_0$ is the initial state, $Act'$ is the set of actions, $\hookrightarrow \,\subseteq S \times Act' \times S$ is the transition relation, described as the following.
			\begin{itemize}
				\item The global state of a Timed Rebeca model $s \in S$ in FTTS is the same as that of in the standard semantics. In comparison with the standard semantics, the values of \textit{program counter} and \textit{resuming time} are set to zero and $\epsilon$, respectively, for all actors in FTTS. As a result,  states of actors in FTTS are in the form of $(v, q, 0, t, \epsilon)$. In addition, there is no guarantee for the local times of actors to be the same, i.e. time \emph{floats} across the actors in the transition system. 
				\item The initial state of a model in FTTS is the same as that of in TTS.
				\item The set of actions is defined as $Act' = \mathit{MName}.$
				\item The transition relation $\hookrightarrow \,\subseteq S \times Act' \times S$ defines the transitions between states that occur as the results of actors' activities including: taking a message from the message box, and executing all of the statements of its corresponding message server. For proposing the formal definition of $\hookrightarrow$, we have to define the notion of \textit{idle} actors. An actor in the state $(v, q, \epsilon, t, r)$ is idle if it is not busy with executing a message server. Consequently, a given state $s$ is idle, if $s(x)$ is idle for every actor $x$. We use the notation $\mathit{idle}(s,x)$ to denote the actor identified by $x$ is idle in state $s$, and $\mathit{idle}(s)$ to denote $s$ is idle. Using these definitions, two states $s, s' \in S$ are in relation $s \longhooktrans{mg} s'$ if and only if the following conditions hold.
				\begin{itemize}
					\item $\mathit{idle}(s) \wedge \mathit{idle}(s')$, and
					\item $\exists\, s_1, s_2, \cdots, s_n \in S, x \in \mathit{AID} \cdot s \longtrans{mg} s_1 \rightarrow \cdots \rightarrow s_n \rightarrow s' \wedge \forall y \in \mathit{AID}/\{x\},~ 1 \leq i \leq n \cdot \neg \mathit{idle}(s_i, x) \wedge \mathit{idle}(s_i, y)$
				\end{itemize}
				
				% \begin{table}
				%  {
				% \begin{center}
				%   	\begin{tabular}{cr}
				%         $\sosruleNormal{8cm}{s(x)=(v, q, \langle \mathbf{delay}(e)|\sigma\rangle, t, r) \wedge r = t} {s \longtrans{\tau} s[x\mapsto (v, q, \sigma, t+\mathit{eval}_v(e), r+\mathit{eval}_v(e))]}$ & \textbf{(delay)}
				% 	\end{tabular}
				% \end{center}
				%         }
				% \end{table}
			\end{itemize}
		\end{defi}
		
		More details and SOS rules which define these transitions in FTTS are presented in \cite{RTS} and \cite{Ehsan:Thesis:2018}.
	}

	%%%%%%%%%%%%%EHSAN%%%%%%%%%%%%
	%Floating Time Transition System (FTTS), introduced in \cite{RTS}, gives a natural event-based semantics for timed actors, providing a significant amount of reduction in the state space. FTTS is a coarse-grain semantics and contains only event transitions of the standard semantics; each state transition shows the effect of handling of a message by a rebec.

	% The states of FTTS are defined based on the local states of rebecs. Note that in Timed Rebeca models, actors have synchronized local clocks which gives the modeler a notion of global time. But in state space generated based on the FTTS semantics, in each state different actors may have different local times and the time \emph{floats} across the actors in the state space \cite{RTS}. In a state, an actor with a message in its message bag,  takes the message and executes its corresponding message server. If the message server has a delay statement, the local time of actor advances independently. 
	
	\begin{example}
		The FTTS of the Rebeca model given in Figure \ref{Fig:ُTimedRebecaModelingExample} is given in Figure \ref{fig::FTTS}. 
		\EK{As mentioned before, the values of resuming times and program counters are set to $\epsilon$ and zero in FTTS, so they are not shown in the figure.} 
		In the initial state, called $t_1$, only the rebec $\it req$ has a message ``$\it request$" in its bag, and the local time of all rebecs is $0$. Upon handling the message ``$\it request$", the local time of rebec $\it req$ is progressed to $3$ and a message ``$\it request$" is inserted to the bag of $\it res$ as shown in the state $t_2$. Upon handling the message ``$\it request$" by rebec $\it res$, as its arrival time is $11$, the local time of rebec is progressed to $11$ in the state $t_3$.
	\end{example}
	
	%Let $\ID$ denote the set of Rebeca identifiers, and $S$ the set of global states. Each global state $s\in S$ is a mapping from the Rebeca identifier to its local state. Assume $\Var$, $\Value$, and $\Msg$ be the set of variables, values, and messages, respectively. We use the notation $\bag(\Msg)$ to represent the bag of messages and $\int$ to denote the local time of actors. So, the set of global states is defined by mapping each rebec identifier to its local state, $S=\ID\rightarrow (\Var\rightarrow \Value)\times \bag(\Msg) \times \int$. Each message $m\in\Msg$ constitute of three parts, namely $m=(\msgsig, \arrival, \deadline)$, where $\msgsig$ is the message content, $\arrival$ is the arrival time of the message, and $\deadline$ is the deadline of the message. We use $\msgsig(m)$, $\arrival(m)$, and $\deadline(m)$ to indicate the corresponding parts. The message content constitutes of the name of message and its parameter values. We use ${\it Type}(\msgsig(m))$ to show the name of the message content. Let $\statevars(s(x))$, $\bag(s(x))$, and $\now(s(x))$ denote the state variable valuation, message bag, and the local time of the rebec with the identifier $x\in\ID$. 
	
	%As the local time of actors progresses, the number of states grows.
	\EK{As proved in \cite{RTS}, the FTTS and the TTS of a given Timed Rebeca model are in a weak bisimulation relation. Hence, the FTTS preserves the timing properties of its corresponding TTS, specified by weak modal $\mu$-calculus where the actions are taking messages from the bag of actors.

		There is no explicit time reset operator in Timed Rebeca; so, the progress of time results in an infinite number of states in transition systems of Timed Rebeca models in both TTS and FTTS. However, Timed Rebeca models generally show periodic or recurrent behaviors, i.e. they perform periodic behaviors over infinite time. Based on this fact, in \cite{FTTS} a new notion for equivalence relation between two states is proposed to make transition systems finite, called \textit{shift equivalence relation}. Intuitively, when building the state space there may be a new state $s$ generated in which the local states of rebecs are the same as an already existing state $s'$, and the new state  $s$ only differs from $s'$ in a fixed shift in the value of parts which are related to the time (the same shift value for all timed related values, i.e. now, arrival times of messages, and deadlines of messages). Such new states can be merged with the older ones to make the transition systems bounded. %This reduction is introduced by an equivalence relation on FTTS states. Recall that each message $m\in\Msg$ constitutes of three parts, namely $m=(\msgsig, \arrival, \deadline)$. We use $\msgsig(m)$, $\arrival(m)$, and $\deadline(m)$ to denote these three elements. We use ${\it Type}(\msgsig(m))$ to show the name of the message. Let $\statevars(s(x))$, $\bag(s(x))$, and $\now(s(x))$ denote the state variable valuation, message bag, and the local time of the rebec with the identifier $x\in\ID$, respectively.

		%Rebecs may have periodic behaviors \FG{and so, the corresponding FTTS of such models are not bounded}. 

		%merges the states $s$ and $s'$ %such that there exists the fixed value $\delta$ and for all $x\in \ID$:
		%that the local time of their rebecs has a fixed delay with each other, called shift equivalent.
		\begin{defi}[shift-equivalence relation]\label{Def::def1}
			Two states $s$ and $s'$ are called \emph{shift equivalent}, denoted by $s\simeq_\delta s'$, if for all the rebecs with identifier $x\in\ID$ there exists $\delta$ such that:
			\begin{enumerate}
				\item %\REM{Condition on state variables:}
				$\statevars(s(x))=\statevars(s'(x))$,
				\item %\REM{Condition on local time:}
				$\now(s(x)) = \now(s'(x))+\delta$,
				\item 
				%{\small
				%\[
				$ \forall m\in \bag(s(x))\Leftrightarrow (\msgsig(m),\arrival(m)+\delta,\deadline(m)+\delta)\in\bag(s'(x)).$
				% \]
				% }
			\end{enumerate}
			
		\end{defi}

		%%%Intuitively, the local time of rebecs in $s'$ has the fixed shift value $\delta$ with respect to the local time of rebecs in $s$. %{In other words, it can be considered $s'$ as a state occurred in future of $s$, but with the same observational behavior, i.e., generating the same events.} \fixme{a state doesnt have a behavior} 
		%%%Intuitively and informally we may say that when we are in $s'$, it is like if sometime in the future we are back in the same state as $s$ (considering all the values of variables, contents of the message bags, (relative) time tags of messages and (relative) deadlines). 
		
		%In shift equivalence relation two states are equivalent if and only if they are the same except for the value of parts which are related to the time (value of now, arrival times of messages, and deadlines of messages) and shifting the value of parts which are related to the time in one state makes it the same as the other one. This way, instead of preserving absolute value of time, only the relative difference of timing parts of states are preserved. We remark that the first and third conditions force the state variables of rebecs and the message names (including message parameters) of corresponding rebecs in the two states be equal \cite{FTTS}.

		The \emph{bounded floating-time transition system}s (BFTTS) $\langle S_f,\hookrightarrow, Act' ,s_{0_f}\rangle$ of a Timed Rebeca model is obtained by merging states of its corresponding FTTS $\langle S,\rightarrow, Act',s_{0} \rangle$ that are in shift equivalence relation. Shift equivalent states are merged into the one that its rebecs have the smallest local times. %\sout{Assume that $s$ is a state of FTTS that its shift equivalent states are merged into a state denoted by $[s]_{\simeq}$, i.e., $\forall s'\in S \cdot s\simeq_\delta s' \Leftrightarrow s'\simeq_{\delta'} [s]_{\simeq}\,\wedge \, \forall x\in\ID\cdot(\now(s(x))\ge \now([s]_\sim(x)))$. Formally speaking, if $([s]_{\simeq},m,s')\in \hookrightarrow $ in BFTTS as a consequence of processing the message $m$, then there exists $s''\in S$ such that $(s,m,s'')\in \rightarrow$ and $s'\simeq_{\delta^\ast} s''$ for some $\delta^\ast.$ Conversely, if $(s,m,s'')\in \rightarrow$ in FTTS as a consequence of processing the message $m$, then there exists $s'\in S_f$ such that $([s]_{\simeq},m,s')\in\hookrightarrow$ and $s'\simeq_{\delta^\ast} s''$ for some $\delta^\ast$}. In other words,} 
		In \cite{RTS} it is proved that FTTS and its corresponding BFTTS are strongly bisimilar; so, BFTTS of a Timed Rebeca model preserves the timing properties of its corresponding FTTS. % specified by weak modal $\mu$-calculus where the actions are taking messages from the bag \cite{RTS}.
		%TODO::check propeties
	}
	
	\begin{example}
		The FTTS of Figure \ref{fig::FTTS} modulo shift-equivalence is partially shown in Figure \ref{fig::BFTTS}. Assume the state $t_6$ in FTTS with the same configuration of $t_6'$ in BFTTS. In the state $t_6$, rebec $\it req$ handles its ``$\it response$" message and as a consequence it sends a ``$\it request$" message to itself and its local clock is advanced to $32$. We call the resulting state $t_7$. The local clocks of rebecs in states $t_7$ and $t_4$ have a $16$-time difference and the values of their state variables and bag contents are equal. So, these two states are shift-equivalent and are merged, resulting the loop in the corresponding BFTTS. 
	\end{example}
	%-------------------------------------------------

	\section{Modeling Patterns in Rebeca}\label{sec::model}
	
	We use the architecture proposed in \cite{7318707} for implementing communication patterns. We will explain the main components of the Publisher-Subscriber pattern as the others are almost the same. 
	As illustrated in Figure \ref{Fig:PubSub}, the pattern provides communication between two application components -- a \emph{client} and a \emph{service},  each of which could be either a software app or a medical device. For example, the client could be a pulse oximeter publishing SPO2 values to a monitoring application service. In our modeling approach, each of the patterns will have a component acting as an interface on either side of the communication that abstracts the lower-level details of a communication substrate. In this case, there is a \textit{PublisherRequester} component that the publisher calls to send a message through the communication substrate and a \textit{SubscriberInvoker} that receives the message from the communication substrate and interacts with the service.
	This structure is common in most communication middleware (e.g., the Java Messaging Service, or OMG's Data Distribution Service) in which APIs are provided to the application level components and then behind the scenes a communication substrate handles marshalling/unmarshalling and moving the message across the network using a particular transport mechanism. % This architecture specifies two interfaces between its constituent roles, e. g., publisher and subscriber,
	% and the communication substrate. 
	% These interfaces encapsulate details of patterns from low-level details of various substrate layers. 
	% As illustrated in Figure \ref{Fig:PubSub}, the client and service are devices/apps which aim to communicate with each other. 
	In our approach to reason about timing properties, the interface components \textit{PublisherRequester} and \textit{SubscriberInvoker} check the local timing properties related to the client or service side, 
	respectively.
	% and the \textit{communication substrate} component is responsible for transmitting data.
	
	\begin{figure}[htbp]
		\centering
		\frame
		{\includegraphics[width=.73\linewidth]{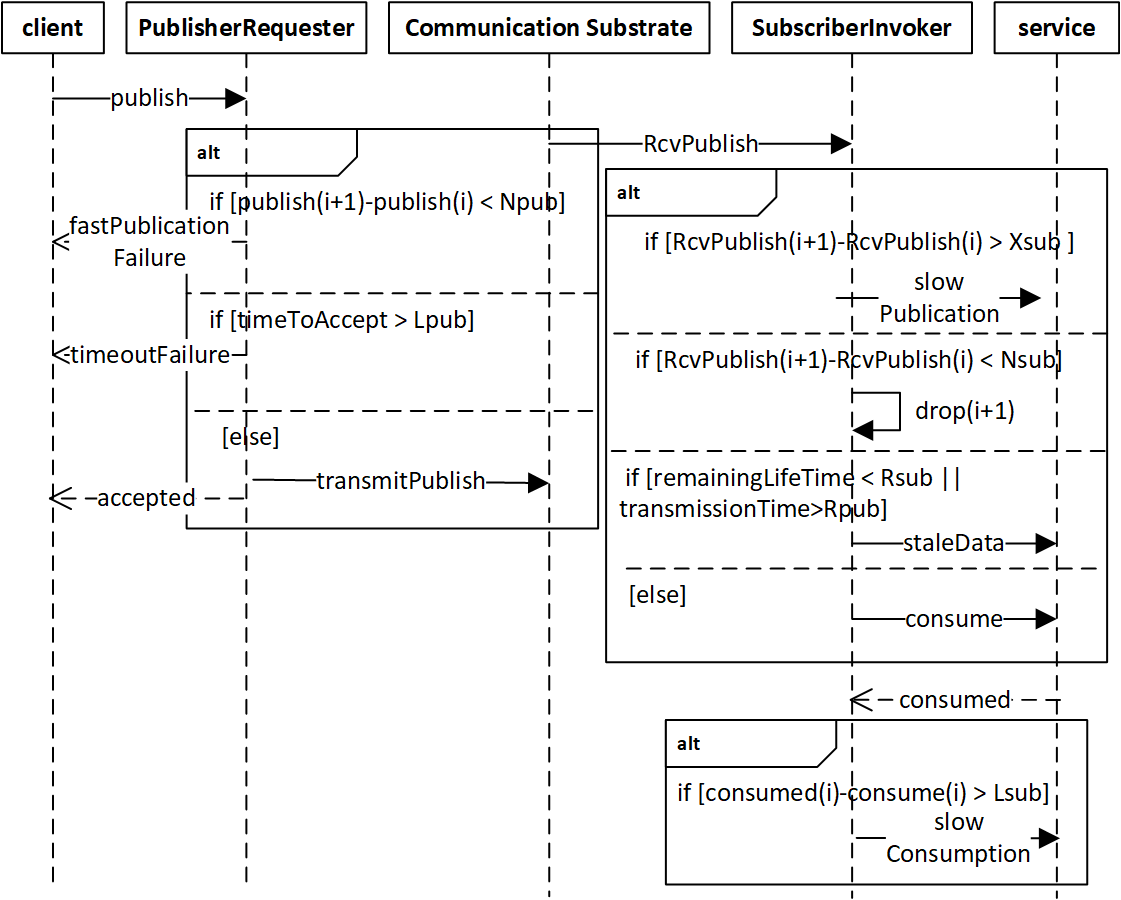}}
		\caption{Publisher-Subscriber pattern sequence diagram: %$N_{pub}$ is the minimum separation required between two consecutive \emph{publish} messages, $L_{pub}$ is the maximum latency that communication substrate can accept a message, $R_{pub}$ is the maximum remaining lifetime that the subscriber can accept a message, $N_{sub}$ is the minimum separation required between two received messages at the subscriber, $X_{sub}$ is the maximum separation between arrival of two consecutive messages at the subscriber, $L_{sub}$ is the maximum latency that the subscriber can consume a message, $R_{sub}$ is the minimum remaining lifetime that the subscriber can accept a message.}
			\FG{see Section \ref{subsec::ppparam} for the definitions of parameters $N_{pub}$, $L_{pub}$, $N_{sub}$, $X_{sub}$, $L_{sub}$, and $R_{sub}$.}
		}
		\label{Fig:PubSub}
	\end{figure}
	
	We model each component of this architecture as a distinct actor or rebec in Rebeca. We explain the model of the Publisher-Subscriber pattern in details. Other patterns are modeled using a similar approach.
	% Other patterns are modeled with the same discussion.
	
	Figure \ref{fig:Publisher} illustrates \textit{PublisherRequester} reactive class, which is the interface between the client (device/app) and the communication substrate.
	As we see in lines $3$ and $4$, it has two known rebecs: the communication substrate \textit{cs} to which messages are forwarded, and the client \textit{c} to which it will return messages indicating the success/failure status of the communication.  
	% The instances of this reactive class can send messages to them.
	We define the state variable $\lastPub$ in line $5$ for saving the time of last publication message. We use this time for computing the interval between two consecutive messages. This rebec has a message server named \textit{publish}.
	We pass \emph{Lm} and \emph{life} parameters through all message servers in the model to compute the delivery time and remaining lifetime of each message. As the communication substrate is impacted by network traffic, the communication delay between the interface and the communication substrate is non-deterministic. To model this communication delay, % delay between the interface and the communication substrate,
	we define the variable \emph{clientDelay} (in line $11$) with non-deterministic values. The parameters of \emph{Lm} and \emph{life} are updated in lines $12$ and $13$ based on \emph{clientDelay}. This interface is responsible for checking $N_{pub}$ and $L_{pub}$ properties as specified in lines $15$-$23$. To check $N_{pub}$, the interval between two consecutive \emph{publish} messages should be computed by subtracting the current local time of rebec from $\lastPub$. The reserved word \emph{now} represents the local time of the rebec. As this reserved word can not be used directly in expressions, we first assign it to the local variable \emph{time} in line $14$. If both properties are satisfied, it sends a \emph{transmitPublish} message to the communication substrate and an \emph{accepted} message to the client. \MZ{The \emph{accepted} message notifies the \emph{publisher} that the \emph{publish} message was sent to the subscriber through the communication substrate}. These messages are delivered to their respective receivers with a non-deterministic delay, modeled by \emph{clientDelay}, using the statement $\after$.
	It means that the message is delivered to the client after passing this time. In case that the $N_{pub}$ property is violated, it sends a message \emph{fastPublicationFailure} to the client. If the $L_{pub}$ property is violated, it sends a message \emph{timeOutFailure}.

	\begin{figure}[htbp]
\begin{lstlisting}[style=customjava,label=Publisher,multicols=2]
reactiveclass PublisherRequester(20){
  knownrebecs{
   CommunicationSubstrate cs;
   Client c;
  }
  statevars{int lastPub;}

  PublishRequester(){lastPub = 0;}

  msgsrv publish(boolean data,int topic,int Lm,int life){
	int clientDelay=?(1,2);
	Lm=Lm+clientDelay;
	life=life-clientDelay;
	int time = now;
	if(time-lastPub<NPUB){
	  c.fastPublicationFailure(data,topic,Lm);
	}
	if(clientDelay>LPUB){
	  c.timeOutFailure(data,topic,Lm);
	}
	else{
	  lastPub = now;
	  cs.transmitPublish(data,topic,Lm,life) after(clientDelay);
	  c.accepted() after(clientDelay);
	}
  }
}
\end{lstlisting}
		\caption{Modeling publisher interface in Timed Rebeca}
		\label{fig:Publisher}
	\end{figure}

	\textit{Communication substrate} abstracts message passing middleware 
	%% From John -- since ethernet, CAN etc, have interesting protocol issues that we are not modeling here, 
	%% I would feel more comfortable just saying that the abstraction is for message-passing middleware.
	%% We could list things like JMS, DDS, etc., but we just listed those a few paragraphs ago.
	%%
	%% a network like Ethernet, wireless networks, Controller Area Network (CAN) bus \cite{pfeiffer2008embedded} 
	by specifying the outcomes of message passing.
	% effects of the network on transmitting messages. 
	To this aim, it may consider priorities among received messages to transmit or assign specific or non-deterministic latency for sending messages. A specification of \textit{communication substrate} reactive class is shown in Figure \ref{fig:CS}. It handles \emph{transmitPublish} messages by sending a $\RcvPublish$ message to its known rebec, a rebec of \textit{SubscriberInvoker} class in line $11$. It considers a non-deterministic communication delay for each message, modeled by the local variable \emph{netDelay} in line $8$. This rebec updates the parameters \emph{Lm} and \emph{lifetime} based on \emph{netDelay} before sending $\RcvPublish$ in lines $9$ and $10$.
	
	\begin{figure}[htbp]
\begin{lstlisting}[style=customjava,label=CS]
reactiveclass CommunicationSubstrate(20){
  knownrebecs{SubscriberInvoker si;}
  CommunicationSubstrate(){}
  msgsrv transmitPublish(boolean data,int topic,int Lm,int life){
	int netDelay=?(1,2);
	Lm=Lm+netDelay;
	life=life-netDelay;
	si.RcvPublish(data,topic,Lm,life) after(netDelay);
  }
}
\end{lstlisting}
		\caption{Modeling communication substrate in Timed Rebeca}
		\label{fig:CS}
	\end{figure}

	The \textit{SubscriberInvoker} reactive class, given in Figure \ref{fig:Subsciber}, is an interface between the communication substrate and the service (device/app). It has only one known rebec that is the destination for the messages of its instances. We define a state variable $\lastPub$ in line $5$ to save the time of the last publication message that arrived in this rebec. This reactive class is responsible for checking $N_{sub}$, $X_{sub}$, $R_{pub}$, $R_{sub}$, and $L_{pub}$ properties (see Subsection \ref{SS:patterns}). Message servers in this rebec are $\RcvPublish$ and $\consume$. 
	% From John -- I suggest removing this sentence because it seems to be repeating what was just said two sentences previous.
	%
	% It checks $N_{sub}$, $X_{sub}$, $R_{pub}$, and $R_{sub}$ properties in the message server $\RcvPublish$. 
	The $\RcvPublish$ server begins by modeling the communication delay between the interface and the service by assigning the variable \emph{serviceDelay} (in line $8$) with a non-deterministic value. \FG{As we explained in Section \ref{SS:patterns}, a subscriber states its need about the timing of incoming data by parameters $N_{sub}$ and $X_{sub}$: the rate at which it consumes data.} It computes the interval between two consecutive $\RcvPublish$ messages in lines $9-10$ and then uses this value to check that it be greater than equal to the minimum separation constraint at line $11$ ($N_{sub}$) and less than equal to the maximum separation constraint at line $14$ ($X_{sub}$). \FG{Otherwise, subscriber concludes a too fast or slow publication, respectively. A subscriber also states its need about the freshness of data by timing properties $R_{pub}$ and $R_{sub}$; if data arrives at the subscriber late after its publication (by comparing \emph{Lm} and $R_{pub}$) or its remaining lifetime is less than $R_{sub}$, subscriber concludes that data is stale (line $17$) and sends a failure message to the service. }
	% It checks  $R_{sub}$ and $R_{pub}$ properties using  \emph{life} and \emph{Lm} parameters at line $17$. Any violation of these properties will result in sending a failure message to the service or dropping the message. 
	By satisfying the properties, it saves the local time of the actor in $\lastPub$ and sends a $\consume$ message to the service $\after$ a delay of \emph{serviceDelay}. Handling the message reception notification from the service, the message server $\consumed$ checks $L_{sub}$ property in line $29$ and sends a failure to the service if the consumption time exceeded the specified maximum consumption latency $L_{sub}$.

	\begin{figure}[htbp]
\begin{lstlisting}[style=customjava,label=Subsciber,multicols=2]
reactiveclass SubscriberInvoker(20){
  knownrebecs{Service s;}
  statevars{int lastPub;}

  SubscribeInvoker(){lastPub = 0;}

  msgsrv RcvPublish(boolean data,int topic,int Lm,int life){
    int serviceDelay=?(1,2);
	int time = now;
	int interval=time-lastPub;
	if(interval<NSUB){
	  self.drop(data,topic,Lm);
	}
	if(interval>XSUB){
	  s.slowPublication(data,topic,Lm) after(serviceDelay);
	}
	if(life<RSUB||Lm>RPUB){
	  s.staleData(data,topic,Lm) after(serviceDelay);
	}
	else{
	  lastPub = now;
	  s.consume(data,topic,Lm+serviceDelay) after(serviceDelay);
	}
  }
  msgsrv consumed(boolean data,int topic,int Lm){
	int time = now;
	int serviceDelay=?(1,2);
	if(time-lastPub>LSUB){
	  s.slowConsumption(data,topic,Lm+ serviceDelay) after(serviceDelay);
	}
  }
  msgsrv drop(boolean data,int topic,int Lm){...}
}
\end{lstlisting}
		\caption{Modeling subscriber interface in Timed Rebeca}
		\label{fig:Subsciber}
	\end{figure}
	
	The \textit{Service} reactive class, given in Figure \ref{fig:Service}, consumes the publish message with a delay and then assigns $L_m$ to $\it transmissionTime$ indicating the point-to-point message deliver latency from publication of data to its receipt.
	
	\begin{figure}[htbp]
\begin{lstlisting}[style=customjava,label=service,multicols=2]
reactiveclass Service(5) {
	knownrebecs{
		SubscriberInvoker si;
	}
	statevars {
		int transmissionTime;
	}
	 Service() {
	 transmissionTime = 0;
	}
	msgsrv consume(boolean data,int topic,int Lm) {
		int consumeDelay = ?(1, 2); 
		delay(consumeDelay);
		transmissionTime = Lm;
		si.consumed(data,topic,Lm + consumeDelay); 
	}
	msgsrv slowConsumption(boolean data,int topic,int Lm){}
	msgsrv staleData(boolean data,int topic,int Lm){}
	msgsrv slowPublication(boolean data,int topic,int Lm){}
}
\end{lstlisting}
		\caption{Modeling service in Timed Rebeca}
		\label{fig:Service}
	\end{figure}
	
	\subsection{Analysis of Patterns}
	
	Checking of local timing properties is encoded in component models, and failure to satisfy such properties is indicated by notification messages sent to relevant components. Non-local timing properties are specified using assertions in the property language of the model checker and are checked during the model checking process.
	% We specify checking of all the local QoS properties in the model, and it affects on the behavior of the model during the state-space generation. 
	% We check the communication QoS requirements by model checking. To this aim, we specify the property of each pattern as an assertion in a separate 
	% property file as shown in Figure
	% \ref{Fig:PubSubProperty}. 
	An example of such an assertion (corresponding to the non-local timing property of inequality \ref{eq:pub-sub}) is shown in Figure~\ref{Fig:PubSubProperty}.
	$\latency$ is the name of property and  $\transmission$ is the maximum message delivery time ($L_M$ in inequality \ref{eq:pub-sub}) that is the state variable of the $\service$ class. By using $\assertion$ keyword, the property should be satisfied in all the states of the model.
	
	Table \ref{Table:modelCheck} shows the result of the analysis of patterns using Rebeca. We assign two groups of values to parameters to show the state-space size for when the timing requirement is satisfied and when it is not satisfied.
	As we see the state-space size is smaller when the timing requirement inequality is not satisfied, because in this situation a path in which the property is not satisfied is found and the state-space generation stops.

	\begin{figure}[h]
\begin{lstlisting}[style=customjava]
property{ 
  Assertion{
    LatencyOverload : ((RPUB - RSUB - LPUB) >= service.transmissionTime);
  }
}
\end{lstlisting}
		\caption{The timing requirement for the Publisher-Subscriber pattern}
		\label{Fig:PubSubProperty}
	\end{figure}
	
	\begin{table}[h]
		\begin{center}
			\caption{Analysis of patterns in Rebeca}
			\begin{tabular}{|c||c|c|c|c|c|c|}
				\hline
				\multirow{2}{*}{\begin{tabular}{c} Communication\\ pattern \end{tabular}}     & \multirow{2}{*}{Parameters}               & \multicolumn{3}{c|}{\multirow{2}{*}{No. states}} & \multicolumn{2}{c|}{\multirow{2}{*}{Timing requirement}} \\
				&       & \multicolumn{3}{c|}{}                               & \multicolumn{2}{c|}{}                               \\ \hline
				\multirow{2}{*}{\begin{tabular}{c} Publisher\\ Subscriber \end{tabular}} & \begin{tabular}[c]{@{}l@{}}$N_p=4,L_p=6$\\ $R_p=40, N_s=7$\\$L_s=5, X_s=20$\\$R_s=5$\end{tabular}  & \multicolumn{3}{c|}{235}                  & \multicolumn{2}{c|}{$R_{pub}-R_{sub}-L_{pub} \geq L_{m}$}                  \\ \cline{2-7}
				&
				\begin{tabular}[c]{@{}l@{}}$N_p=5,L_p=3$\\ $R_p=20, N_s=4$\\$L_s=7,X_s=12$\\$R_s=10$\end{tabular}                                              &  \multicolumn{3}{c|}{56}                  & \multicolumn{2}{c|}{$R_{pub}-R_{sub}-L_{pub} \ngeq L_{m}$}                  \\ \hline
				\multirow{2}{*}{\begin{tabular}{c}  Requester\\Responder \end{tabular}} & \begin{tabular}[c]{@{}l@{}}$N_r=2,L_r=30$\\ $R_r=15,N_s=7$\\$L_s=5, R_s=10$\end{tabular}                                              & \multicolumn{3}{c|}{205}           & \multicolumn{2}{c|}{$L_{req}+R_{req}-L_{res}-R_{res} \geq L_{m}+L'_{m}$}                  \\ \cline{2-7}
				& \begin{tabular}[c]{@{}l@{}}$N_r=4,L_r=24$\\ $R_r=18, N_s=5$\\$L_s=10, R_s=20$\end{tabular}                                                & \multicolumn{3}{c|}{113}   & \multicolumn{2}{c|}{$L_{req}+R_{req}-L_{res}-R_{res} \ngeq L_{m}+L'_{m}$}                  \\ \hline
				\multirow{2}{*}{\begin{tabular}{c} Sender\\ Receiver \end{tabular}} & \begin{tabular}[c]{@{}l@{}}$N_s=5,L_s=30$\\ $N_r=7,L_r=5$\end{tabular}                                              & \multicolumn{3}{c|}{179}   & \multicolumn{2}{c|}{$L_{sen}-L_{rec} \geq L_{m}+L'_{m}$}                  \\ \cline{2-7}
				&  \begin{tabular}[c]{@{}l@{}}$N_s=4,L_s=6$\\ $N_r=7,L_r=5$\end{tabular}                                               & \multicolumn{3}{c|}{82}                  & \multicolumn{2}{c|}{$L_{sen}-L_{rec} \ngeq L_{m}+L'_{m}$}                  \\ \hline
				\multirow{2}{*}{\begin{tabular}{c} Initiator\\ Executor \end{tabular}} & \begin{tabular}[c]{@{}l@{}}$N_i=3,L_i=25$\\ $N_e=5,L_e=4$\end{tabular}                                              & \multicolumn{3}{c|}{169}   & \multicolumn{2}{c|}{$L_{ini}-L_{exe} \geq L_{m}+L'_{m}$}                  \\ \cline{2-7}
				&  \begin{tabular}[c]{@{}l@{}}$N_i=5,L_e=10$\\ $N_i=5,L_e=4$\end{tabular}                                               & \multicolumn{3}{c|}{49}                  & \multicolumn{2}{c|}{$L_{ini}-L_{exe} \ngeq L_{m}+L'_{m}$}               \\ \hline
			\end{tabular}
			\label{Table:modelCheck}
		\end{center}
	\end{table}
	
	As we assumed that apps/devices satisfy their timing constraints, violation of each inequality shows a run-time failure in which the communication substrate fails to communicate the data quickly enough to meet the interface components timing requirements and so the network is not configured properly. Alternatively, it can be seen as a design error where, e.g., the interface components are making too stringent local timing requirements (e.g., on the freshness).
	
	\subsection{Guidelines on Modeling Composite Medical Systems}\label{subsec::guide}
	Depending on the configuration of a composite medical system, devices and applications connect to each other through a specific pattern. For each connection of two devices/apps, two interface components are needed as summarized in Table \ref{Tab::guideline}. 
	
	\begin{table}[htbp]
		\centering
		\caption{Interface components and communicated messages %, and variables measuring the intervals between consecutive messages 
			for each pattern }
		\label{Tab::guideline}
		\begin{tabular}{|c|c|c|}
			\hline
			Pattern     &  Interface  &  Comm.  \\%$& Interval \\ 
			& Components & Message \\%& Variable\\
			\hline
			\multirow{2}{*}{Publisher-Subscriber}   & \textit{PublisherRequester} & \textit{transmitPublish} \\%& $\it lastPub$\\
			& \textit{SubscriberInvoker} & \\ % & $\it lastPub$\\
			\hline
			\multirow{2}{*}{Request-Responder} & \textit{RequestRequester} & \textit{transmitRequest}\\% & $\it lastReq$ \\
			& \textit{ResponderInvoker} & \textit{transmitResponse}\\% & $\it lastReq$\\ 
			\hline
			\multirow{2}{*}{Initiator-Executor} & \textit{InitiatorRequester} & \textit{transmitInitiate}\\% & $\it lastInit$\\
			& \textit{ExecutorInvoker} & \textit{transmitAck} \\ % & $\it lastInit$\\ 
			\hline
			\multirow{2}{*}{Sender-Receiver} & \textit{SenderRequester} & \textit{transmitSend} \\% & $\it lastSend$\\
			& \textit{ReceiverInvoker} & \\%& $\it lastSend$\\
			\hline
		\end{tabular}
	\end{table}

	In a composite medical system, there may be device/apps that communicate over a shared message passing middleware. In such cases, we should also share the \textit{communication substrate} among the corresponding patterns of device/apps. It is a design decision to be faithful to the patterns (and the system). \FG{As a shared \textit{communication substrate} communicates with all interface components of involved patterns, then we have to pass these components via its constructor in our models. Instead, to make the specification of a shared \textit{communication substrate} independent from its interface components, we use the inheritance concept in Rebeca. We implement a base reactive class for the shared \textit{communication substrate} and all interface components as shown in Figure \ref{fig:Base}, inspired by the approach of  \cite{yousefi2019verivanca}. 
		%We implement a base reactive class for the \textit{communication substrate} of patterns as shown in Figure \ref{fig:Base} named \textit{Base}, inspired by the approach of  \cite{yousefi2019verivanca}. 
		We define the state variable $\it id$ in line $2$ to uniquely identify rebecs. This class has a lookup method named \emph{find} to get the rebec with a given identifier as its parameter. Thanks to the special statement \emph{getAllActors} (in line $4$), we can get access to all actors extending the \textit{Base}. In the method \emph{find}, we define an array of reactive classes and initiate it by calling \emph{getAllActors}. We iterate over the actors of this array to find the actor with the given identifier (in lines $5-7$).} %get ids of all actors that are derived from the \textit{Base} (in line $6$) actor and search through them for finding the specified one (line $7$).
	
	\begin{figure}[h]
\begin{lstlisting}[style=customjava]
reactiveclass Base(20){
  statevars {int id;}
  Base find(int _id){
    ArrayList<ReactiveClass>  allActors getAllActors();
    for(int i = 0 ; i < allActors.size(); i++){
    Base actor = (Base) allActors.get(i);
    if (actor.id == _id) {return actor;}
  }}}
\end{lstlisting}
		\caption{Base Reactive Class}
		\label{fig:Base}
	\end{figure}

	The \textit{communication substrate} reactive class \textit{extends} \textit{Base} class. As illustrated in Figure \ref{fig:csb}, this class has a parameter $\it id$ in its constructor for assigning the $\it id$ variable of the parent class (in line $3$). This class has no known rebec as opposed to the one specified in Figure \ref{fig:CS}. Instead, rebecs append their identifiers to their messages during their communication with the substrate. The communication substrate class uses the \textit{find} method for finding the rebec that wants to send data based on its id  (lines $6$ and $11$).  %\sout{Depending on the patterns that uses the common communication substrate,} 
	\FG{\textit{Communication substrate} class includes a message server for each communicated message of all patterns as shown in Table \ref{Tab::guideline} (in addition to those for error messages). We can remove the message servers of unused patterns for the sake of readability. However, there is no cost if we do not remove the additional ones as when an event is not triggered, it will not be handled. This class specification can be used as a template even when we have no sharing.}  As \textit{communication substrate} class in Figure \ref{fig:csb} is commonly used by the components of Publisher-Subscriber and Requester-Responder patterns, it has the two message servers \textit{transmitPublish} and \textit{transmitRequest} for Publisher-Subscriber, and \textit{transmitResponse} for Requester-Responder.

	\begin{figure}[htbp]
\begin{lstlisting}[style=customjava]
reactiveclass CommunicationSubstrate extends Base(20){

  CommunicationSubstrate(int _id){id = _id;}
  
  msgsrv transmitPublish(boolean data,int topic,int Lm,int life,int subscriberId){
    int csDelay = ?(1, 2);
    SubscribeInvoker si=(SubscribeInvoker) find(subscriberId);
    si.publish(data,topic,Lm+csDelay,life-csDelay) after(csDelay);
  }
  msgsrv transmitRequest(boolean data,int Lm,int responderInvokerId){
    int cs1Delay =  ?(1, 2);
    ResponderInvoker ri=(ResponderInvoker) find(responderInvokerId);
    ri.request(data,Lm + cs1Delay) after(cs1Delay);
  }
  msgsrv transmitResponse(boolean data,int Lm,int life,int requesterId) {
    int cs2Delay =  ?(1, 2);
    RequestRequester rr = (RequestRequester) find(requesterId);
    rr.response( data, (Lm + cs2Delay), (life-cs2Delay)) after(cs2Delay);
  }
}
\end{lstlisting}
		\caption{Modeling Communication substrate using inheritance}
		\label{fig:csb}
	\end{figure}

	All interface components that communicate through a shared \textit{communication substrate} should also extend the \textit{Base} class. \FG{For each usage of a pattern, one instance of component interfaces on both side is needed. When an interface component is instantiated, the identifier of its counterpart interface component is set via the constructor. When an interface component sends a message to its counterpart interface component via our proposed communication substrate, it includes the identifier of the counterpart entity.} 
	
	%\REM{with minor changes to the proposed models in Section \ref{sec::model} (use the id of service/client)} 
	
	Each device/app may use several patterns to communicate with other device/app. Depending on its role in each pattern, we consider a known rebec of the appropriate interface component. To model devices/apps, we only focus on the logic for sending messages through the interface components.
	
	\FG{We consider three types of delays in our specifications; 1) between sender interface components and communication substrate (\textit{after} in the sender interface in the Rebeca model), and 2) between communication substrate and receiver interface components (\textit{after} in the communication substrate in the Rebeca model), 3) between receiver interface component to receiver application (\textit{after} in the receiver interface in the Rebeca model). 
		
		The sender interface component models the behavior of communication driver, and the first type of delay defined in this component models the delay caused by network traffic. The driver retries until it successfully sends its message depending on the traffic. The value of this delay is defined by \emph{clientDelay}  in the sender interface component.
		The second type of delay, defined by \emph{netDelay} in communication substrate, shows the delay of message passing middleware on transferring messages (for example caused by the routing and dispatching algorithms).
		The third type of delay, defined by \emph{serviceDelay} in the receiver interface components, shows the delay caused by the system load.
		When a receiver interface receives messages, it should send them to the application components. Depending on the system load, the operating system allows the interface components to deliver their message to the application components.
		
		%The sender interface component models the behavior of communication driver. The driver retries until it successfully sends its message depending on the traffic. So, the first type of delays, its value defined by \emph{clientDelay} in sender interface components, models the delay caused by network traffic. 
		%The second type of delay, defined by \emph{netDelay} in communication substrate, shows the delay of message passing middleware on transferring messages. When a receiver interface receives messages, it should send them to the application components. Depending on the system load, the operating system allows the interface components to deliver their message to the application components. So, the third delay, defined by \emph{serviceDelay} in the receiver interface components, shows the delay caused by the system load.
	}

	%-------------------------------------------------
	\section{State-space Reduction}\label{sec::reduction}
	A medical system is composed of several devices/apps that communicate with each other by using any of communication patterns. With the aim of verifying the timing requirements of medical systems before deployment, we use the Rebeca model checking tool, Afra. As we explained in Section \ref{sec::model}, each communication pattern is at least modeled by five rebecs. It is well-known that as the number of rebecs increases in a model, the state space grows significantly. For a simple medical system composed of two devices that communicate with an app, there are nine rebecs in the model (as communication substrate is in common). In a more complex system, adding more devices may result in state-space explosion, and model checking cannot be applied. We propose a partial order reduction technique for FTTS which merges those states satisfying the same local timing properties of communication patterns. The reduced model preserves the class of timed properties specified by weak modal $\mu$-calculus of the original model where the actions are taking messages from the message bag \cite{RTS}. 
	
	%In our models, we have only considered one round of communication over the patterns between two device/apps. If the behaviors of device/app were recurrent while 
	If we use the value of $\now$ in our Rebeca codes, then it is very likely that we encounter an unbounded state space, because the first condition of shift-equivalent relation given in Definition \ref{Def::def1} may not be satisfied. By this condition, two states are shift-equivalent if all state variables of all rebecs have the same value. %Recall that the first condition is applied on the state variables of rebecs.
		Here, we suggest a restricted form of using the value of $\now$ %to measure the time interval between two execution points of a rebec. We define this restricted form of usage 
		by specifying a set of variables called \emph{interval variables}, % that are used for measuring the interval between \FG{two execution points of a rebec. 
		%Though using such variables does not still result in satisfaction of the first condition. 
		%
		and we relax the first condition %of Definition \ref{Def::def1} 
		for such variables.
		%to bounded state spaces.% and the message contents in the bags.
		The model checking tool of Rebeca, Afra, is adjusted to treat these variables differently and hence we prevented generating an unbounded state space.
		
		%We consider those state variables that are used for measuring the interval between \FG{two execution points of a rebec. For instance, the state variable $\lastPub$ of Rebeca class \emph{PublisherRequeter} in Publisher-Subscriber pattern is used within the expression ``$\now-\lastPub$'' to measure the interval between two consecutive $\it publish$ messages. This interval value is used to } check local timing properties like $N_{sub}$, $N_{pub}$, and $X_{sub}$. % and the value of $\lastPub$ is not used anymore.  
		%Such variables grow as the local time of rebecs proceeds. 
		Let $\Interval$ be the set of interval variables %, called \emph{interval variables}, 
		that are defined to only hold time values, i.e., we can only assign the value of $\now$ to these variables at different points of the program. We use these variables to measure the time interval between two execution points of a rebec by comparing the value of $\now$ at these two points. At the first point of execution, we assign $\now$ to the interval variable $x\in\Interval$, and at the second point the expression ``$\now-x$''  measures the time interval between the first and the second point. For instance, the state variable $\lastPub$ of Rebeca class \emph{PublisherRequeter} in Publisher-Subscriber pattern is used to % within the expression ``$\now-\lastPub$'' to 
		measure the interval between two consecutive $\it publish$ messages. This interval value is used to check local timing properties like $N_{sub}$, $N_{pub}$, and $X_{sub}$. %\FG{For any $x\in\Interval$, they are only used in expressions as $\now-x$ as right-value to measure an interval, or we only assign $\now$ as a left-value (for maintaining a specific time).} 
		% 
		%By relaxing the first condition of Definition \ref{Def::def1}, two states are shift-equivalent if all state variables of all rebecs have the same value except for interval variables. 
		As variables of $\Interval$ can only get the value of $\now$, we relax the first condition of Definition \ref{Def::def1} on these variables, and we treat such variables similar to the local time (see Section \ref{subsec::FTTS}).% \sout{This idea can be applied to similar variables measuring an interval in other types of classes. We have summarized the name of these variables for each pattern in Table \ref{Tab::guideline}. We assume that these variables names are reserved and are not used in other classes. Let $\Interval$ be the set of pattern variables measuring the interval between two consecutive messages. }}

We also relax the third condition of Definition \ref{Def::def1} which compares the messages in rebec bags. \rev{The local timing properties may impose restrictions on the transmission time or the remaining life time of data. % of data communicated between two device/app to ensure its freshness. %lower and upper bounds on the time that the message transfer between two components. 
In our implementations, we model the transmission time or remaining life time of data by the parameters of messages, namely $L_m$ and \emph{life}, respectively (see the message server $\RcvPublish$ in the Rebeca class \emph{SubscriberInvoker}). % in messages to model the remaining life time. %We decrease this time by the delay based on the delay when pass it to the next recbec. 
To check the timing property on the freshness of data, the parameter \emph{life} is compared with the constant $R_{sub}$, configured as the system parameter, in a conditional statement. This is the only place that the variable is used within its message server. % and hence, the value of this variable has no effect on the future behavior of its rebec. 
So, we can abstract the concrete value of this parameter and only consider the result of the comparison. Instead of passing $\life$ as the parameter, we can pass the Boolean result of $\life<R_{sub}$ instead upon sending the message $\RcvPublish$. We use this interpretation to compare the messages with such parameters in bags. Instead of comparing the values of parameters one-by-one, for the parameters similar to \emph{life} we only consider the result of the comparison. If the values of \emph{life} for two $\RcvPublish$ messages either both satisfy or violate the condition $\life<R_{sub}$ while other parameters are equal, these messages can be considered equivalent as both result the same set of statements to be executed (irrespective to the value of this parameter). This idea can also be applied to the message $\it RcvResponse$ in the Request-Responder pattern, parameterized by \emph{life} which is compared to the local timing property $R_{res}$ in its message server. Formally speaking,} we identify those messages, denoted by $\Msg_{ex}$, whose parameters are only used in conditional statements (if-statements). We use the conditions in the if-statements for data abstraction. We can abstract the concrete values of the parameters and only consider the result of evaluation of conditions. %parameters either both satisfy or violate the condition $\life<R_{sub}$, these messages can be considered equivalent as both result the same set of statements to be executed (irrespective to the value of this parameter)
		%We distinguish a set of messages %We consider those messages in bags that the values of some of their parameters only affect the execution of its corresponding message servers. 
%\rev{\sout{For example, the Rebeca class \emph{SubscriberInvoker} has the message $\RcvPublish$ parameterized with \emph{life}. The value of this parameter is only used by its message server within a conditional statement to check the condition $\life<R_{sub}$. This variable is not used anymore and hence, the value of this variable has no effect on the future behavior of its rebec. If the values of this parameter for two $\RcvPublish$ messages either both satisfy or violate the condition $\life<R_{sub}$, these messages can be considered equivalent as both result the same set of statements to be executed (irrespective to the value of this parameter). This idea can also be applied to the message $\it RcvResponse$ in the Request-Responder pattern, parameterized by \emph{life} which is compared to the local timing property $R_{res}$ in its message server.}} %For a set of Rebeca classes, we identify those messages, denoted by $\Msg_{ex}$ whose parameters that their values are only checked within specific conditions and have no effect on the future behavior of its rebecs. 
In Definition \ref{Def::def2} (relaxed-shift equivalence relation), for simplicity, we assume that each message $m\in\Msg_{ex}$ has only one parameter that we abstract its value, denoted by ${\it msgpar}_{ex}(m)$, %and is used by its message server to check only one 
based on the result of one condition, denoted as ${\it cond}(m)$. We also assume that the message has another parameter that its value is denoted by ${\it msgpar}_{\overline{ex}}(m)$. By ${\it cond}(m)({\it msgpar}_{\it ex}(m))$, we mean the result of evaluation of the condition that is checked by the message server $m$ over its specific parameter of $m$. The concrete values of a message parameter can be abstracted if it is only used in the conditions of if-statements. The concrete value of a parameter is needed if it is used in other statements (e.g., assignment or send statement in Timed Rebeca).  \rev{To ensure the soundness of our abstraction, we limit ${\it cond}(m)$ to propositional logic over the comparison of ${\it msgpar}_{\it ex}$ with constants. Considering more complicated conditions is among of our future work.} In practice, we can find such parameters and their corresponding conditions through a static analysis of the message server or ask the programmer to identify them.

		\begin{defi}[relaxed-shift equivalence relation]\label{Def::def2}
			Two semantic states $s$ and $s'$, denoted by $s\sim_{\delta}s'$, are relaxed-shift equivalent if there exists $\delta$ such that for all the rebecs with identifier $x\in\ID$ :
			\begin{enumerate}
				
				\item {\small\[\begin{array}{l}\forall v\in\Var\setminus{\Interval} 
					%\{\lastPub,\lastReq\}
					\cdot\statevars(s(x))(v) =\statevars(s'(x))(v),\\
					\forall v\in \Interval\cap{\it Dom}(s(x))\Rightarrow\statevars(s(x))(v) = \statevars(s'(x))(v)+\delta. \end{array}\]}
				%\lastPub\in{\it Dom}(s(x))\Rightarrow\statevars(s(x))(\lastPub) = \statevars(s'(x))(\lastPub)+\delta, \\
				%\lastReq\in{\it Dom}(s(x))\Rightarrow\statevars(s(x))(\lastReq) = \statevars(s'(x))(\lastReq)+\delta. \end{array}\]}
				
				\item  {\small$\now(s(x)) = \now(s'(x))+\delta$}.
				
				\item  {\small \[\begin{array}{l}\forall m\in \bag(s(x)) \,\wedge %{\it Type}(\msgsig(m))
					m\not\in \Msg_{ex} \Leftrightarrow %\\    %\hspace*{1.5cm} 
					(\msgsig(m),\arrival(m)+\delta,\deadline(m)+\delta)\in\bag(s'(x)),\vspace*{1mm}\\ \forall m\in \bag(s(x)) \,\wedge %{\it Type}(\msgsig(m))
					m\in \Msg_{ex} \Leftrightarrow %\\ \hspace*{1cm} 
					\exists m'\in\bag(s'(x)) \cdot %\, \wedge \, 
					{\it Type}(\msgsig(m))={\it Type}(\msgsig(m')) \,\wedge \\\hspace*{1cm} 
					\arrival(m') =\arrival(m)+\delta \,\wedge \, \deadline(m')= \deadline(m)+\delta \,\wedge \\\hspace*{1.2cm} {\it msgpar}_{\overline{ex}}(m)={\it msgpar}_{\overline{ex}}(m') \, \wedge \,  ({\it cond}(m)({\it msgpar}_{\it ex}(m)) \Leftrightarrow {\it cond}(m')({\it msgpar}_{ex}(m'))).
					\end{array}\]}
			\end{enumerate}
			
		\end{defi}

		We consider Timed Rebeca models that in their message servers $\now$ can be only accessed for updating variables of $\Interval$ or used in expressions like ``$\now-x$'' (where $x\in\Interval$) for computing an interval value. %In other words, $x$ and $\now$, i.e., variables related to time, are not used within expressions for assigning a value to a state variable or parameters of messages in send statements.} 
		We reduce FTTSs of such models by merging states that are relaxed-shift equivalent. The following theorem shows that an FTTS modulo relaxed-shift equivalence preserves the properties of the original one.

		\begin{thm}\label{Th::preserve}
			For the given FTTS $\langle S,s_0,\rightarrow\rangle$, assume the states $s,s'\in S$ such that $s\sim_{\delta}s'$. If $(s,m,s^\ast)\in\rightarrow$, then there exists $s^{\ast\ast}$ such that $(s',m,s^{\ast\ast})\in\rightarrow$ and $s^\ast\sim_{\delta}s^{\ast\ast}$.
		\end{thm}
		
		\begin{proof}
			Assume an arbitrary transition $(s,m,s^\ast)\in\rightarrow$ that denotes handling the message $m$ by the rebec $i$. Based on the third condition of Definition \ref{Def::def2}, %there also exists a message $m'$ such that ${\it Type}(\msgsig(m'))={\it Type}(\msgsig(m))$ in the bag of Rebec $i$ in the state $s'$. Assume $s^{\ast\ast}$ is the resulting state as the consequence of handling $m'$ in the state $s'$. We prove that $s^\ast\sim_{\delta}s^{\ast\ast}$. Regarding ${\it Type}(\msgsig(m))$, 
			two cases can be distinguished:
			\begin{itemize}
				\item $m\not\in \Msg_{ex}$: there exists message $m'\in\bag(s'(i))$ such that $m'= (\msgsig(m),\\\arrival(m)+\delta,\deadline(m)+\delta)$. %\in\bag(s'(x))$. 
				The local time of rebec $i$ in both states $s$ and $s'$ are progressed by executing delay statements in the message servers $m$ and $m'$ which is the same and hence, their local timers have still $\delta$-difference in the resulting states $s^\ast$ and $s^{\ast\ast}$. So, the second condition is satisfied (result $\dagger$). Due to the execution of the same message server, all variables are updated by the same expressions. Based on the constraint on the models and the assumption $s\sim_{\delta}s'$, we conclude that all the variables except for $\Interval$ have the same values (defined by the same expressions on un-timed variables). The values of $\Interval$ are updated to $\now$ and, by the result $\dagger$, they still have $\delta$-difference. %We remark that all variables of $\Interval$ %except $\{\lastPub,\lastReq\}$ 
				%may be accessed/updated during execution of the message handler. 
				%Due to the execution of the same message server, all variables except of $\Interval$ are updated by the message handler $m$ and $m'$ to the same expressions %We remark that the variables $x\in\Interval$ are only accessed in expressions like $\now-x$ which are equal. If they are updated, they will be updated to $\now$, and by the result $\dagger$, they will still have $\delta$-difference.  
				%As rebec $i$ has only access to its own variables, the variables of other rebecs do not change. Thus, handling the message $m$ and $m'$ in 
				So, $s^{\ast}$ and $s^{\ast\ast}$ satisfy the first condition. The messages sent to other rebecs during handling $m$ and $m'$ are sent at the same point with the same parameters (defined by expressions on untimed-variables). As their local timers have $\delta$-difference, the arrival and deadline of sent messages have $\delta$-difference. So, the third condition is also satisfied.
				\item $m\in \Msg_{ex}$: there exists message $m'\in\bag(s'(i))$ with the same message name of $m$, i.e., ${\it Type}(\msgsig(m))$,  arrival time $\arrival(m)+\delta$, and deadline $\deadline(m)+\delta$. The first assumption of Definition \ref{Def::def2} together with the assumptions $ {\it msgpar}_{\overline{ex}}(m)={\it msgpar}_{\overline{ex}}(m')$ and $({\it cond}(m)({\it msgpar}_{\it ex}(m)) \Leftrightarrow {\it cond}(m')({\it msgpar}_{ex}(m')))$ result that the same statements are executed by the rebec $i$ during handling $m$ and $m'$. The variables $x\in\Interval$ are used in expressions as $\now-x$ and so  $\now(s(i))-\statevars(s(i))(x) = \now(s'(i))-\statevars(s'(i))(x)$. With a similar discussion in the previous case, the local time of rebec $i$ in both states $s$ and $s'$ are progressed by executing the same delay statements and hence, their local timers have still $\delta$-difference in the resulting states $s^\ast$ and $s^{\ast\ast}$. So, the second condition is satisfied (result $\ddagger$). Due to the execution of the same statements, all variables are updated by the same expressions. Based on the constraint on the models and the assumption $s\sim_{\delta}s'$, we conclude that all the variables except of $\Interval$ have the same values (defined by the same expressions on un-timed variables). The values of $\Interval$ are updated to $\now$ and, by the result $\ddagger$, they still have $\delta$-difference. So, $s^{\ast}$ and $s^{\ast\ast}$ satisfy the first condition. The messages sent to other rebecs during handling $m$ and $m'$ are sent at the same point with the same parameters (defined by expressions on untimed-variables). As their local timers have $\delta$-difference, the arrival and deadline of sent messages have $\delta$-difference. So, the third condition is also satisfied. \qedhere
			\end{itemize}
		\end{proof}
	
	\begin{cor}\label{Co::preserve}
		For the given FTTS $\langle S,s_0,\rightarrow\rangle$, assume the states $s,s'\in S$ such that $s\sim_{\delta}s'$. Then, $s$ and $s'$ are strongly bisimilar.
	\end{cor}
	\begin{proof}
		To show that $s$ and $s'$ are strongly bisimilar, we construct a strong bisimulation relation $R$ such that $(s,s')\in R$. Construct $R=\{(t,t')
		\mid t\sim_{\delta}t'\}$. We show that $R$ satisfies the transfer conditions of strong bisimilarity. For an arbitrary pair $(t,t')$, we must show that\begin{itemize}
			\item If $(t,m,t^\ast)\in \rightarrow$, then there exists $t^{\ast\ast}$ such that $(t',m,t^{\ast\ast})\in \rightarrow$ and $(t^\ast,t^{\ast\ast})\in R$; 
			\item If $(t',m,t^{\ast\ast})\in \rightarrow$, then there exists $t^{\ast}$ such that $(t,m,t^\ast)\in \rightarrow$ and $(t^\ast,t^{\ast\ast})\in R$.
		\end{itemize}
		If $(t,m,t^\ast)\in \rightarrow$, by Theorem \ref{Th::preserve}, there exists $t^{\ast\ast}$ such that $(t',m,t^{\ast\ast})\in \rightarrow$ and $t^\ast\sim_{\delta}t^{\ast\ast}$. By the construction of $R$, $t^\ast\sim_{\delta}t^{\ast\ast}$ implies that $(t^\ast,t^{\ast\ast})\in R$. The same discussion holds when $(t',m,t^{\ast\ast})\in \rightarrow$. Concluding that $R$ is a strong bisimulation. Trivially $(s,s')\in R$.  
	\end{proof}

	The relaxed-shift equivalence preserves the conditions of shift equivalence on all variables except the time related variables, i.e., $\Interval$. %$\{\lastPub,\lastReq\}$. 
	Furthermore, it preserves the conditions of shift equivalence on all messages in the bag except for messages $\Msg_{ex}$. But the condition on parameters of ${\it msgpar}_{\it ex}$, i.e.,  $({\it cond}(m)({\it msgpar}_{\it ex}(m)) \Leftrightarrow {\it cond}(m')({\it msgpar}_{ex}(m')))$ ensures that the same statements of the corresponding message server will be executed. Therefore, by Corollary \ref{Co::preserve} any FTTS modulo relaxed-shift equivalence is strongly bisimilar to its original FTTS, and it not only preserves the local timing properties (those properties checked on variables of $\Interval$ in model specification like $\now-\lastPub<N_{\it pub}$) of the original one but also preserves the timing properties defined on events (taking messages from the bag).

	\begin{example}
		Figure \ref{Fig:bftts} shows part of the FTTS for the Publisher-Subscriber pattern. As we see all local times and the values of lastPub are shifted one unit in states $s_{17}$ and $s_{20}$, so the first and the second conditions of Definition \ref{Def::def2} are satisfied. The message contents of their bags are equal in all rebecs except for the rebec $si$. This rebec has a $\it RcvPublish$ message in its bag in both states that their $\life$ values are different but both of them are greater than the $R_{sub}$ value. So the third condition of Definition \ref{Def::def2} is satisfied too and the states are relaxed-shift equivalent and we can merge them. Respectively, we can merge states $s_{22}$ with $s_{25}$ and $s_{27}$ with $s_{30}$. We remark that these states are not merged in the corresponding BFTTS. %The reduced FTTS is shown in Figure \ref{Fig:relaxed}.
	\end{example}
	
	\begin{figure}[htbp]
		\centering
		\tikzstyle{line} = [draw, -stealth]
\tikzstyle{block} = [rectangle, draw, text width=20em, text centered, rounded corners, minimum height=2em, node distance = 12em]
\tikzstyle{first} = [rectangle, draw, text width=8em, text centered, rounded corners, minimum height=2em, node distance = 15em]
\tikzstyle{dotted} = [draw, densely dotted, -stealth]
\begin{tikzpicture}[scale=0.45, transform shape]
\node[block](s16)
{
	\begin{tabularx}{\textwidth}{l|l|X} 
	\multicolumn{3}{c}{$s_{16}$}       \\ 
	\hline
	\multirow{2}{*}{\rotatebox[origin=c]{90}{c}} & bag  & []  \\ 
	\cline{2-3}
	& state variables &  \\
	\cline{2-3}
	& now & 8 \\ 
	\hline
	\multirow{2}{*}{\rotatebox[origin=c]{90}{pr}} & bag  & []  \\ 
	\cline{2-3}
	& state variables & lastPub=6 \\
	\cline{2-3}
	& now & 8 \\ 
	\hline
	\multirow{2}{*}{\rotatebox[origin=c]{90}{cs}} & bag  & [(publish(8,28),8,$\infty$)]  \\ 
	\cline{2-3}
	& state variables &  \\ 
	\cline{2-3}
	& now & 8 \\
	\hline
	\multirow{2}{*}{\rotatebox[origin=c]{90}{si}} & bag  & []  \\ 
	\cline{2-3}
	& state variables & lastPub=0 \\ 
	\cline{2-3}
	& now & 8 \\
	\hline
	\multirow{2}{*}{\rotatebox[origin=c]{90}{s}} & bag  & []  \\ 
	\cline{2-3}
	& state variables & \\ 
	\cline{2-3}
	& now & 8 \\
	\end{tabularx}
};
%\\\\\\\\\\\\\\\\\\\\\\\\\\\\\\\\\\\\\\\\\\\\\\\
\node[block, xshift=22em](s17)
{
	\begin{tabularx}{\textwidth}{l|l|X} 
	\multicolumn{3}{c}{$s_{17}$}       \\ 
	\hline
	\multirow{2}{*}{\rotatebox[origin=c]{90}{c}} & bag  & [(accepted(),8,$\infty$)]  \\ 
	\cline{2-3}
	& state variables &  \\
	\cline{2-3}
	& now & 8 \\ 
	\hline
	\multirow{2}{*}{\rotatebox[origin=c]{90}{pr}} & bag  & []  \\ 
	\cline{2-3}
	& state variables & lastPub=6 \\
	\cline{2-3}
	& now & 8 \\ 
	\hline
	\multirow{2}{*}{\rotatebox[origin=c]{90}{cs}} & bag  & []  \\ 
	\cline{2-3}
	& state variables &  \\ 
	\cline{2-3}
	& now & 8 \\
	\hline
	\multirow{2}{*}{\rotatebox[origin=c]{90}{si}} & bag  & [(RcvPublish(8,27),8,$\infty$)]  \\ 
	\cline{2-3}
	& state variables & lastPub=0 \\ 
	\cline{2-3}
	& now & 8 \\
	\hline
	\multirow{2}{*}{\rotatebox[origin=c]{90}{s}} & bag  & []  \\ 
	\cline{2-3}
	& state variables & \\ 
	\cline{2-3}
	& now & 8 \\
	\end{tabularx}
};
%\\\\\\\\\\\\\\\\\\\\\\\\\\\\\\\\\\\\\\\\\\\\\\\
\node[block, below of=s16, yshift=-15em, xshift=12em](s22)
{
	\begin{tabularx}{\textwidth}{l|l|X} 
	\multicolumn{3}{c}{$s_{22}$}       \\ 
	\hline
	\multirow{2}{*}{\rotatebox[origin=c]{90}{c}} & bag  & []  \\ 
	\cline{2-3}
	& state variables &  \\
	\cline{2-3}
	& now & 8 \\ 
	\hline
	\multirow{2}{*}{\rotatebox[origin=c]{90}{pr}} & bag  & []  \\ 
	\cline{2-3}
	& state variables & lastPub=6 \\
	\cline{2-3}
	& now & 8 \\ 
	\hline
	\multirow{2}{*}{\rotatebox[origin=c]{90}{cs}} & bag  & []  \\ 
	\cline{2-3}
	& state variables &  \\ 
	\cline{2-3}
	& now & 8 \\
	\hline
	\multirow{2}{*}{\rotatebox[origin=c]{90}{si}} & bag  & [(RcvPublish(8,27),8,$\infty$)]  \\ 
	\cline{2-3}
	& state variables & lastPub=0 \\ 
	\cline{2-3}
	& now & 8 \\
	\hline
	\multirow{2}{*}{\rotatebox[origin=c]{90}{s}} & bag  & []  \\ 
	\cline{2-3}
	& state variables & \\ 
	\cline{2-3}
	& now & 8 \\
	\end{tabularx}
};
%\\\\\\\\\\\\\\\\\\\\\\\\\\\\\\\\\\\\\\\\\
\node[block, below of=s22, yshift=-15em](s27)
{
	\begin{tabularx}{\textwidth}{l|l|X} 
	\multicolumn{3}{c}{$s_{27}$}       \\ 
	\hline
	\multirow{2}{*}{\rotatebox[origin=c]{90}{c}} & bag  & []  \\ 
	\cline{2-3}
	& state variables &  \\
	\cline{2-3}
	& now & 9 \\ 
	\hline
	\multirow{2}{*}{\rotatebox[origin=c]{90}{pr}} & bag  & []  \\ 
	\cline{2-3}
	& state variables & lastPub=6 \\
	\cline{2-3}
	& now & 9 \\ 
	\hline
	\multirow{2}{*}{\rotatebox[origin=c]{90}{cs}} & bag  & []  \\ 
	\cline{2-3}
	& state variables &  \\ 
	\cline{2-3}
	& now & 9 \\
	\hline
	\multirow{2}{*}{\rotatebox[origin=c]{90}{si}} & bag  & []  \\ 
	\cline{2-3}
	& state variables & lastPub=9 \\ 
	\cline{2-3}
	& now & 9 \\
	\hline
	\multirow{2}{*}{\rotatebox[origin=c]{90}{s}} & bag  & [(publish(9),9,$\infty$)]  \\ 
	\cline{2-3}
	& state variables & \\ 
	\cline{2-3}
	& now & 9 \\
	\end{tabularx}
};
%\\\\\\\\\\\\\\\\\\\\\\\\\\\\\\\\\\\\\\\\\
\path [line] (s16) -- node[xshift=-6em] {(publish(8,28),8,$\infty$)} (s22);
\path [line] (s17) -- node[xshift=-4em] {(accepted(),8,$\infty$)} (s22);
\path [line] (s22) -- node[xshift=-6em] 
{(RcvPublish(8,27),8,$\infty$)} (s27);
%\\\\\\\\\\\\\\\\\\\\\\\\\\\\\\\\\\\\\\\\\\\\\\\\\
\node[block, xshift=45em](s6)
{
	\begin{tabularx}{\textwidth}{l|l|X} 
	\multicolumn{3}{c}{$s_{6}$}       \\ 
	\hline
	\multirow{2}{*}{\rotatebox[origin=c]{90}{c}} & bag  & []  \\ 
	\cline{2-3}
	& state variables &  \\
	\cline{2-3}
	& now & 9 \\ 
	\hline
	\multirow{2}{*}{\rotatebox[origin=c]{90}{pr}} & bag  & []  \\ 
	\cline{2-3}
	& state variables & lastPub=7 \\
	\cline{2-3}
	& now & 9 \\ 
	\hline
	\multirow{2}{*}{\rotatebox[origin=c]{90}{cs}} & bag  & [(publish(9,29),9,$\infty$)]  \\ 
	\cline{2-3}
	& state variables &  \\ 
	\cline{2-3}
	& now & 9 \\
	\hline
	\multirow{2}{*}{\rotatebox[origin=c]{90}{si}} & bag  & []  \\ 
	\cline{2-3}
	& state variables & lastPub=0 \\ 
	\cline{2-3}
	& now & 9 \\
	\hline
	\multirow{2}{*}{\rotatebox[origin=c]{90}{s}} & bag  & []  \\ 
	\cline{2-3}
	& state variables & \\ 
	\cline{2-3}
	& now & 9 \\
	\end{tabularx}
};
%\\\\\\\\\\\\\\\\\\\\\\\\\\\\\\\\\\\\\\\\\\\\\\\
\node[block, xshift=67em](s20)
{
	\begin{tabularx}{\textwidth}{l|l|X} 
	\multicolumn{3}{c}{$s_{20}$}       \\ 
	\hline
	\multirow{2}{*}{\rotatebox[origin=c]{90}{c}} & bag  & [(accepted(),9,$\infty$)]  \\ 
	\cline{2-3}
	& state variables &  \\
	\cline{2-3}
	& now & 9 \\ 
	\hline
	\multirow{2}{*}{\rotatebox[origin=c]{90}{pr}} & bag  & []  \\ 
	\cline{2-3}
	& state variables & lastPub=7 \\
	\cline{2-3}
	& now & 9 \\ 
	\hline
	\multirow{2}{*}{\rotatebox[origin=c]{90}{cs}} & bag  & []  \\ 
	\cline{2-3}
	& state variables &  \\ 
	\cline{2-3}
	& now & 9 \\
	\hline
	\multirow{2}{*}{\rotatebox[origin=c]{90}{si}} & bag  & [(RcvPublish(8,28),9,$\infty$)]  \\ 
	\cline{2-3}
	& state variables & lastPub=0 \\ 
	\cline{2-3}
	& now & 9 \\
	\hline
	\multirow{2}{*}{\rotatebox[origin=c]{90}{s}} & bag  & []  \\ 
	\cline{2-3}
	& state variables & \\ 
	\cline{2-3}
	& now & 9 \\
	\end{tabularx}
};
%\\\\\\\\\\\\\\\\\\\\\\\\\\\\\\\\\\\\\\\\\\\\\\\
\node[block, below of=s6, yshift=-15em, xshift=12em](s25)
{
	\begin{tabularx}{\textwidth}{l|l|X} 
	\multicolumn{3}{c}{$s_{25}$}       \\ 
	\hline
	\multirow{2}{*}{\rotatebox[origin=c]{90}{c}} & bag  & []  \\ 
	\cline{2-3}
	& state variables &  \\
	\cline{2-3}
	& now & 9 \\ 
	\hline
	\multirow{2}{*}{\rotatebox[origin=c]{90}{pr}} & bag  & []  \\ 
	\cline{2-3}
	& state variables & lastPub=7 \\
	\cline{2-3}
	& now & 9 \\ 
	\hline
	\multirow{2}{*}{\rotatebox[origin=c]{90}{cs}} & bag  & []  \\ 
	\cline{2-3}
	& state variables &  \\ 
	\cline{2-3}
	& now & 9 \\
	\hline
	\multirow{2}{*}{\rotatebox[origin=c]{90}{si}} & bag  & [(RcvPublish(8,28),9,$\infty$)]  \\ 
	\cline{2-3}
	& state variables & lastPub=0 \\ 
	\cline{2-3}
	& now & 9 \\
	\hline
	\multirow{2}{*}{\rotatebox[origin=c]{90}{s}} & bag  & []  \\ 
	\cline{2-3}
	& state variables & \\ 
	\cline{2-3}
	& now & 9 \\
	\end{tabularx}
};
%\\\\\\\\\\\\\\\\\\\\\\\\\\\\\\\\\\\\\\\\\
\node[block, below of=s25, yshift=-15em](s30)
{
	\begin{tabularx}{\textwidth}{l|l|X} 
	\multicolumn{3}{c}{$s_{30}$}       \\ 
	\hline
	\multirow{2}{*}{\rotatebox[origin=c]{90}{c}} & bag  & []  \\ 
	\cline{2-3}
	& state variables &  \\
	\cline{2-3}
	& now & 10 \\ 
	\hline
	\multirow{2}{*}{\rotatebox[origin=c]{90}{pr}} & bag  & []  \\ 
	\cline{2-3}
	& state variables & lastPub=7 \\
	\cline{2-3}
	& now & 10 \\ 
	\hline
	\multirow{2}{*}{\rotatebox[origin=c]{90}{cs}} & bag  & []  \\ 
	\cline{2-3}
	& state variables &  \\ 
	\cline{2-3}
	& now & 10 \\
	\hline
	\multirow{2}{*}{\rotatebox[origin=c]{90}{si}} & bag  & []  \\ 
	\cline{2-3}
	& state variables & lastPub=10 \\ 
	\cline{2-3}
	& now & 10 \\
	\hline
	\multirow{2}{*}{\rotatebox[origin=c]{90}{s}} & bag  & [(publish(10),10,$\infty$)]  \\ 
	\cline{2-3}
	& state variables & \\ 
	\cline{2-3}
	& now & 10 \\
	\end{tabularx}
};
%\\\\\\\\\\\\\\\\\\\\\\\\\\\\\\\\\\\\\\\\\
\node[first, above of=s16, xshift=33em](s1){$s_1$};
%\\\\\\\\\\\\\\\\\\\\\\\\\\\\\\\\\\\\\\\\\
\path [line] (s6) -- node[xshift=-6em] {(publish(9,29),9,$\infty$)} (s25);
\path [line] (s20) -- node[xshift=-4em] {(accepted(),9,$\infty$)} (s25);
\path [line] (s25) -- node[xshift=-6em] 
{(RcvPublish(8,28),9,$\infty$)} (s30);
\path [dotted] (s1) -| (s16);
\path [dotted] (s1) -- (s17);
\path [dotted] (s1) -- (s6);
\path [dotted] (s1) -| (s20);
\path [dotted] (s27) -- (4.6,-27);
\path [dotted] (s30) -- (22,-27);
\end{tikzpicture}
		\caption{Part of FTTS for the publisher-subscriber pattern: $c$ is an instance of $\it Client$, $\it pr$ of $\it PublisherReqester$, $\it cs$ of $\it CommunicationSubstrate$, $\it si$ of $\it SubscriberInvoker$, and $\it s$ of $\it Service$.}
		\label{Fig:bftts}
	\end{figure}
	
	%\begin{figure}[htbp]
	%	\centering
	%	\input{./Fig/relaxedftts.tex}
	%	\caption{Relaxed FTTS of Figure \ref{Fig:bftts}}
	%	\label{Fig:relaxed}
	%\end{figure}

	%-------------------------------------------------
	\section{Case Studies}\label{sec::caseStudy}
	Our reduction techniques are more applicable when using several patterns and devices, as one might find in an interoperable medical system. %First we provide guidelines on how a composite medical systems can be specified in terms of the models of patterns. Then, 
	We applied our techniques on the following two interoperability scenarios, modeled based on the guidelines given in Section \ref{subsec::guide}. The first case scenario relies on the Initiator-Executor pattern while the second one uses the Publisher-Subscriber and Request-Responder patterns.
	
	\subsection{X-Ray and Ventilator Synchronization Application}
	\label{sec:x-ray-app}
	
	As summarized \linebreak in~\cite{hatcliff2012rationale}, a simple example of automating clinician workflows via cooperating devices addresses problems in acquiring accurate chest X-ray images for patients on ventilators during surgery~\cite{Langevin-al:AJRCCM99}. To keep the lungs' movements from blurring the image, doctors must manually turn off the ventilator for a few seconds while they acquire the X-ray image, but there are risks in inadvertently leaving the ventilator off for too long.   For example, Lofsky documents a case where a patient death resulted when an anesthesiologist forgot to turn the ventilator back on due to a distraction in the operating room associated with dropped X-ray film and a jammed operating table~\cite{Lofsky:APSFNewsletter05}.
	These risks can be minimized by automatically coordinating the actions of the X-ray imaging device and the ventilator. Specifically, a centralized automated coordinator running a pre-programmed coordination script can use device data from the ventilator over the period of a few respiratory cycles to identify a target image acquisition point where the lungs will be at full inhalation or exhalation (and thus experiencing minimal motion). At the image acquisition point, the controller can pause the ventilator, activate the X-ray machine to acquire the image, and then signal the ventilator to ``unpause'' and continue the respiration. An interoperable medical system realizing this concept was first implemented in  \cite{xRay}.
	
	% To capture a sharp x-ray image from a patient during surgery, we should be ensured that the patient's chest and abdomen are not moving. 
	% Hence, the anesthesia machine ventilator hooked to the patient should be stopped. 
	% For coordinating devices, a third device called supervisor is needed to trigger them appropriately. 
	% The supervisor reads messages from the ventilator and send the appropriate message to the x-ray machine
	% \cite{xRay}.
	
	\begin{figure}[h]
		\centering
		\tikzstyle{block} = [rectangle, draw, 
text width=5em, text centered, rounded corners, minimum height=2em, node distance = 8em]
\tikzstyle{line} = [draw, -stealth, thick]
\begin{tikzpicture}[scale=0.8, transform shape]
%nodes
\node [block] (supervisor) {Controller};
\node [block, right of=supervisor] (ventilator) {Ventilator};
\node [block, left of=supervisor] (x-ray) {X-Ray};
%edges
\path [line] (ventilator) -> (supervisor);
\path [line] (x-ray) -> (supervisor);
\path [line] (supervisor) -> (ventilator);
\path [line] (supervisor) -> (x-ray);
\end{tikzpicture}
		\caption{Communication between entities in the X-Ray and Ventilator Synchronization Application.\label{Fig::x-rayApp}}
	\end{figure}
	
	We model the system assuming that the image acquisition point is identified. Controller initiates starting and stopping actions on two devices through the initiator-executor pattern. 
	% , and by using controller, we have the orchestration 
	% pattern as described in subsection \ref{SS:patterns}. 
	We define two instances of Initiator-Executor in the model for communication of the controller with ventilator and X-ray machine as shown in Figure \ref{Fig:x-rayMain}. Each instance of the pattern needs one instance of \textit{InitiatorRequester} and \textit{ExecutorInvoker} classes as the interface between device/app and the communication substrate as explained in Table \ref{Tab::guideline}. The model of controller is given in Figure \ref{Fig::Controller}. The controller communicates with the communication substrate via {\small{$\it IR\_VENTILATOR$}} and {\small$\it IR\_X\_RAY$}, instances of \textit{InitiatorRequester}. First the controller initiates a stop command to ventilator in its initialization in line $13$. Upon receiving a successful acknowledgement through $\it ack$ message, it initiates a start command to X-ray in line $18$. Upon successful completion of the start command, informed via $\it ack$, it initiates a stop to X-ray in line $21$ and then a start command to ventilator in line $23$. Upon receiving an unsuccessful acknowledge or successful completion of the last command, the controller is terminated by sending a $\it terminate$ message to itself.

	\begin{figure}[h]
\begin{lstlisting}[style=customjava,label=x-rayMain]
main {
	Controller c(ir_v, ir_x):(0);
	InitiatorRequester ir_v(cs):(1, 0, 2);
	ExecutorInvoker ei_v(cs):(2, 3, 1);
	Ventilator v(ei_v):(3);
	InitiatorRequester ir_x(cs):(4, 0, 5);
	ExecutorInvoker ei_x(cs):(5, 6, 4);
	X_Ray x(ei_x):(6);
	CommunicationSubstrate cs():(7);
}
\end{lstlisting}
		\caption{Main part of X-ray and ventilator application}
		\label{Fig:x-rayMain}
	\end{figure}

	\begin{figure}[h]
\begin{lstlisting}[style=customjava,label=controller]
reactiveclass Controller extends Base(20) {
	knownrebecs {
		InitiatorRequester IR_VENTILATOR;
		InitiatorRequester IR_X_RAY;
	}
	statevars {}
	Controller(int _id) {
	  id = _id;
	  self.initiate() after(3);}
	  
	msgsrv initiate() {
	  int iniDelay = ?(1, 2, 3);
	  IR_VENTILATOR.initiate(STOP_VENTILATOR,iniDelay+3) after(iniDelay);
	}
	msgsrv ack(int action, int actionStatus, int Lm) {
	  int ackDelay = ?(1, 2, 3);
	  if(action==STOP_VENTILATOR && actionStatus==SUCCEED) {
		IR_X_RAY.initiate(START_X_RAY, ackDelay) after(ackDelay);
	  }else if(action==START_X_RAY && actionStatus==SUCCEED) {
		delay(2);
		IR_X_RAY.initiate(STOP_X_RAY,ackDelay) after(ackDelay+2);
	  }else if(action==STOP_X_RAY && actionStatus==SUCCEED) {
		IR_VENTILATOR.initiate(START_VENTILATOR,ackDelay) after(ackDelay);
	  }else {self.terminate();}
	}
	msgsrv terminate() {self.terminate();} 
}
\end{lstlisting}
		\caption{The controller model in Timed Rebeca}
		\label{Fig::Controller}
	\end{figure}

	As illustrated in Figure \ref{Fig:x-rayCS}, the \textit{communication substrate} extends the \textit{Base} class for transmitting the two messages of these interface components. % by their ids.

	\begin{figure}[h]
\begin{lstlisting}[style=customjava,label=x-rayCS]
reactiveclass CommunicationSubstrate extends Base(20) {
	CommunicationSubstrate(int _id) {
		id = _id;
	}
	msgsrv transmitInitiate(int action, int Lm, int invokerId) { 
		int csDelay = ?(1, 2, 3); 
		ExecutorInvoker ei = (ExecutorInvoker) find(invokerId);
		ei.initiate(action, Lm + csDelay) after(csDelay);
		}
	msgsrv transmitAck(int action, int actionStatus, int Lm, int initiatorId) {
		int cs2Delay = ?(1, 2); 
		InitiatorRequester ir = (InitiatorRequester) find(initiatorId);
		ir.ack(action , actionStatus, Lm + cs2Delay) after(cs2Delay);
	}
     ...
}
\end{lstlisting}
		\caption{X-Ray and ventilator communication substrate}
		\label{Fig:x-rayCS}
	\end{figure}
	
	\subsection{PCA Safety Interlock Application}
	\label{sec:pca-app}
	
	A Patient-Controlled Analgesia (PCA) pump is a medical device often used in clinical settings to intravenously infuse pain killers (e.g.,  opioids) at a programmed rate into a patient's blood stream. A PCA pump also includes a button that can be pushed by the patient to receive additional bolus doses of drug -- thus allowing patients to manage their own pain relief. PCA infusion is often used for pain relief when patients are recovering from an operation. Despite settings on the pump that limit the total amount of drug infused per hour and that impose
	lock out intervals between each bolus dose, there is still a risk of overdose when using PCA pumps.  
	
	Symptoms of opioid overdose include respiratory depression in which a patient's blood oxygenation (SPO2 as can be measured by pulse oximetry) drops and expelled carbon dioxide (End-Tidal CO2 as can be measured by capnography) increases. A PCA pump alone has no way of telling if a patient is suffering from respiratory depression.   However, using emerging interoperable medical system approaches that leverage MAP infrastructure, a pump can be integrated with a pulse oximeter (to measure blood oxygenation) and a capnometer (to measure ETCO2) and an additional control logic in a monitoring application as shown in Figure~\ref{Fig::monitoringApp}. The monitoring application looks for drops in SPO2 and increases in ETCO2, and if monitored values indicate that respiratory depression may be occuring, the application sends a command to the PCA Pump to halt infusion.   Other signals (not shown in Figure~\ref{Fig::monitoringApp}) may be used to alert care-givers of a problem.
	
	This scenario has been considered in a number of demonstrations in medical device interoperability research (see \eg \cite{Arney-al:ICCPS10,king:sehc10}), in interoperability risk management \cite{Hatcliff-al:ISPCE2018}, and is a subject of current standardization activities. The specifics of the model considered here are inspired by the prototype of Ranganath \footnote{\url{https://bitbucket.org/rvprasad/clinical-scenarios}} that uses OMG's DDS message-passing middleware as the communication substrate.

	% For recovering a patient from an operation, he is controlled by a fixed dose of analgesia connected to an infusion pump. In addition, 
	% he is hooked up to a pulse oximeter to measure his pulse rate and oxygen saturation (SPO2) and to a capnometer to measure the concentration of carbon dioxide in 
	% his respiratory gases (end-tidal co2[ETCO2]) and respiratory rate.
	% A monitoring application is composed of the pulse oximeter, capnometer, and infusion pump as shown in Figure \ref{Fig::monitoringApp} 
	% to control the activation of the infusion pump based on the measurements of the devices.
	% If the application detects any deterioration in the patient's condition, it will deactivate the infusion pump and alert the nurses\footnote{\url{https://bitbucket.org/rvprasad/clinical-scenarios}}.
	
	\begin{figure}[htbp]
		\centering
		\tikzstyle{block} = [rectangle, draw, 
text width=7em, text centered, rounded corners, minimum height=4em, node distance = 8em]
\tikzstyle{line} = [draw, -stealth, thick]
\begin{tikzpicture}[scale=0.7, transform shape]
%nodes
\node [block] (capnometer) {Capnometer};
\node [block, below of= capnometer] (oximeter) {Oximeter};
\node [block, right of = capnometer, xshift = 5em, yshift = -4em] (monitor) {Monitoring Application};
\node [block, right of = monitor, xshift = 5em] (pump) {Pump Infusion};
%edges
\path [line] (capnometer) -| node[yshift=0.6em, xshift=-4em] {ETCO2}(monitor);
\path [line] (capnometer) -| node[yshift=-0.6em, xshift=-4em] {Respiratory Rate}(monitor);
\path [line] (oximeter) -| node[yshift=0.6em, xshift=-4em] {SPO2} (monitor);
\path [line] (oximeter) -| node[yshift=-0.6em, xshift=-4em] {Pulse Rate} (monitor);
\path [line] (monitor) -- node[yshift=0.6em] {Command} (pump);
\end{tikzpicture}
		\caption{Communication between entities in the monitoring application.\label{Fig::monitoringApp}}
	\end{figure}
	
	The capnometer and oximeter devices publish data through the Publisher-Subscriber pattern, and the monitoring application detects if data strays outside of the valid range and sends the appropriate command to disable pump infusion. The model of monitoring application is given in Figure \ref{Fig::Monitor}. Monitor, as the role of service in Publisher-Subscriber, consumes data published by capnometer and oximeter in its $\it consume$ message server. It communicates with these devices via {\small$\it SI\_c$} and {\small$\it SI\_o$} which are instances of \textit{SubscriberInvoker}. Upon receiving a consume message and detecting invalid values of SPO2 or ETCO2, it sends an inactive command to pump in lines $22$ and $29$. As monitoring app communicates with pump via the Requester-Responder pattern, it sends its commands to pump via {\small $\it RR\_p$}, an instance of \textit{RequestRequester}. We abstractly model the invalid/valid values for SPO2 or ETCO2 by $\it false$/$\it true$ values for the parameter $\it data$ of $\it consume$ message server. This parameter together with $\it topic$ models the published data of devices.
	
	\begin{figure}[htbp]
		\centering
\begin{lstlisting}[style=customjava,label=monitor]
reactiveclass MonitoringApp extends Base(20){
	knownrebecs{
		SubscriberInvoker SI_c;
		SubscriberInvoker SI_o;
		RequestRequester RR_p; 
	}
	statevars {
	int transmissionTime;
	int TotalLatency;
	}
	MonitoringApp(int _id) {
	 	 id = _id;
		 transmissionTime = 0;
		 TotalLatency = 0;
	}
	msgsrv consume(boolean data,int topic,int Lm){
		int consumeDelay = ?(1, 2, 3); 
		transmissionTime = Lm;
		if(topic==OX_SPO2 || topic==OX_PULSERATE){
			SI_o.consumed(data,topic,Lm) after(consumeDelay);
			if(data==false){
				RR_p.request(INACTIVE,3) after(consumeDelay+3);}
			else{	
				RR_p.request(ACTIVE,3) after(consumeDelay+3);}
		}
		if(topic==CAP_ETCO2 || topic==CAP_RESPIRATORYRATE){
			SI_c.consumed(data,topic,Lm) after(consumeDelay);
			if(data==false){
				RR_p.request(INACTIVE,3) after(consumeDelay+3);}
			else{	
				RR_p.request(ACTIVE,3) after(consumeDelay+3);}
	}}
	msgsrv response(boolean data,int Lm,int requesterId){
		TotalLatency = Lm;
		self.busy();}
}
\end{lstlisting}
		\caption{The Timed Rebeca model of monitoring application.\label{Fig::Monitor}}
	\end{figure}

	As two devices (capnometer and oximeter) send data by using the Publisher-Subscriber pattern to the monitoring app, there are two instances of the Publisher-Subscriber pattern in the final model. The pump and monitoring app communicate via the Requester-Responder pattern. In the resulting Timed Rebeca model of the application, we define two instances of \textit{PublisherRequester} and \textit{SubscriberInvoker} interfaces in \textit{main} and one instances of \textit{RequestRequester} and \textit{ResponderInvoker}, as shown in Figure \ref{fig:main}. The instance of \textit{CommunicationSubstrate} class shown in Figure \ref{fig:csb}, called \textit{cs}, is used by all the components to send their messages and it includes four message servers for transmitting the messages of these two patterns. %that send their message to Communication Substrate in the patterns.

	\begin{figure}[H]
\begin{lstlisting}[style=customjava]
main{
	Capnometer c(pr_c):(0);
	PublishRequester pr_c(cs):(1, 0, 2);
	SubscribeInvoker si_c():(2, 10);
	Oximeter o(pr_o):(5);
	PublishRequester pr_o(cs):(6, 5, 7);
	SubscribeInvoker si_o():(7, 10);
	CommunicationSubstrate cs():(12);
	MonitoringApp ma(si_c, si_o, rr):(10);
	RequestRequester rr(cs):(11,10,13);
	ResponderInvoker ri(cs):(13, 14, 11);
	Pump p(ri):(14);
}
\end{lstlisting}
		\caption{Main part of PCA safety interlock application}
		\label{fig:main}
	\end{figure}

	\subsection{Communication Substrate Models}\label{SS:network models}
	Different network settings such as Wireless, CAN, Ethernet, etc have different behaviors on transmitting the messages. The latency of the underlying network organization interacts with the local properties of devices namely $R_{\it pub}$, $R_{\it sub}$, and $L_{\it pub}$. To ensure that network latency will not exceed the derived upper bounds in terms of the local timing properties of devices before deployment, %validate the interaction among the QoS properties of components and network, 
	we model the middleware behavior on transmitting messages within the communication substrate. Using the verification results, we can adjust the network by dynamic network configuration or capacity planning in organizations. %The behavior of networks may affects on the local QoS properties. To evaluate the effect of different network settings on the satisfaction of local QoS properties, we specify each network model by an appropriate class $\it CommunicationSubstrare$. 
	\FG{We specify a communication substrate class for each shared network settings based on our proposed template in Section \ref{subsec::guide} and adapt the delay values on transmitting messages and priorities on handling messages. We make an instance from each class for those device/apps that communicate over its corresponding shared network.} We consider three different network settings and hence, three communication substrate models for the PCA safety interlock application: 
	
	\begin{itemize}
		\item The first model shown in Figure \ref{fig:csb} imposes a non-deterministic delay on transmitting the messages of both patterns. The number and possible values for the delays are the same for all messages of both patterns. 
		\item The second model considers the same number (which is also equal to the first model) but different values for the delays.% different delay values for patterns.
		\item In the third model as shown in Figure \ref{Fig:CaseStudyCS2}, the network handles messages of the Publisher-Subscriber pattern with a higher priority over the messages of Requester-Responder. This is implemented by using \textit{@priority} statement before their message servers. % and in the main part as shown in Figure \ref{Fig:mainWithPriority}. 
		Communication substrate handles messages based on their arrival time. Messages that are arrived at the same time, are handled based on their priorities. A lower value indicates to a higher priority. %Rebecs are also executed based on By giving priority to rebec instances in the main block,  The instance of the communication substrate has the lowest priority in order to receive all messages from all other patterns and then process them regarding to their priorities.
		%% I should remove this as priority at the level of actors has no effect
	\end{itemize}

	\begin{figure}[h]
\begin{lstlisting}[style=customjava,label=CaseStudyCS2]
reactiveclass CommunicationSubstrate extends Base(20){
    CommunicationSubstrate(int _id){id = _id;}
    @priority(1)
    msgsrv trasmitPublish(int Lm,int life,int subscriberId){...}
    @priority(2)
    msgsrv transmitRequest(int Lm,int responderInvokerId){...}
    @priority(2)
    msgsrv transmitResponse(int Lm,int life,int requesterId){...}
}
\end{lstlisting}
		\caption{Communication substrate model with applying priority to patterns}
		\label{Fig:CaseStudyCS2}
	\end{figure}

	%\begin{figure}[h]
	%	\lstinputlisting[style=customjava,label=mainWithPriority]{./codes/mainWithPriority.rebeca}
	%	\caption{Main Part of Monitoring Application Rebeca Model with Patterns Priority}
	%	\label{Fig:mainWithPriority}
	%\end{figure}

	\subsection{Experimental Results}
	\FG{We extended Afra which applies our reduction technique during the state-space derivation of a complete given model. The tool adds a state to the set of previously generated states if it is not relaxed-shift equivalent to any of them. This on-the-fly application of reduction during state-space generation results in an efficient memory consumption. }\rev{Our tool currently does not support the third condition of Definition~\ref{Def::def2} and we have hard-coded the comparison on messages in the state-space generator for the case study: If the message is of $\RcvPublish$ type, we compare the result of $\life<R_{sub}$ instead of comparing the value of $\life$. } 
	%We developed a code in Java which automatically reduces the resulting FTTSs of these models generated by Afra.
	
	We applied our reduction technique on the model of some patterns and case studies\footnote{The Rebeca models %and the Java code for the reduction of semantic models 
		are available at \url{fghassemi.adhoc.ir/shared/MedicalCodes.zip}}. We got $23\%$ reduction for Requester-Responder, $32\%$ for Publisher-Subscriber, $7\%$ for Initiator-Executor and \MZ{and $8\%$ for Sender-Receiver}. The Initiator-Executor \MZ{and The Sender-Receiver} pattern only have variables measuring the interval between two consecutive messages while the Requester-Responder and Publisher-Subscriber patterns also have remaining lifetime parameter in their messages (\emph{life}) for which our reduction technique relaxes the merge condition. So, the first two patterns have more reduction as their states may reduce with the first and third conditions of Definition \ref{Def::def2}, but the states of Initiator-Executor \MZ{and Sender-Receiver} only reduce with the first condition. In the PCA Monitoring Application which is a medical system using several patterns as explained in Section \ref{sec::caseStudy} we have $29\%$ reduction in the state space and for the X-ray and ventilator application we have $27\%$ reduction.
	
	\begin{table}[h]
		\begin{center}
			\caption{Reduction in patterns and their composition}
			\begin{tabular}{|c||c|c|c|}
				\hline
				\textbf{Model}
				&
				\textbf{No. states} & \textbf{No. states} &
				\textbf{Reduction}\\
				& \textbf{before reduction} & \textbf{after reduction} & \\
				\hline
				Requester-Responder & $205$ & $157$ & $23\%$\\
				Publisher-Subscriber  & $235$ & $159$ & $32\%$ \\
				Initiator-Executor & $113$ & $103$ & $7\%$\\
				\MZ{Sender-Receiver} & $179$ & $164$ & $8\%$\\
				Monitoring App. & $1058492$ & $753456$ &  $29\%$ \\
				X-Ray-Ventilator & $27755$ & $19309$ &  $27\%$ \\
				\hline
			\end{tabular}
			\label{tab:reduction}
		\end{center}
	\end{table}

	Table \ref{Table:caseStudyReduction} shows the reduction of the three network models described in Section \ref{SS:network models}. As we see the reduction in the first model is $28.82\%$. \FG{As the possible delay values are increased for one pattern in the second model, the state space size grows and the reduction increases to $28.92\%$.} In the prioritized model, the resulting state space is smaller than the others. After applying the reduction approach, the state space size reduces to $23.9\%$. %\FG{\sout{The reason for a smaller percent of reduction is the smaller state space. In other words, %the network model is independent of the reduction of state space. As we see 
	for larger state spaces we have more reduction, hence in more complicated systems with more components, we will have a significant amount of reduction in the state space to analyze the system more easily.%}}
	
	\begin{table}[h]
		\begin{center}
			\caption{Reduction of PCA Monitoring Application with different communication substrates}
			\begin{tabular}{|c||c|c|c|}
				\hline
				\textbf{\begin{tabular}{c} Communication \\substrate model \end{tabular}}
				& \textbf{\begin{tabular}{c}  No. states \\before reduction \end{tabular}}
				& \textbf{\begin{tabular}{c}  No. states \\after reduction \end{tabular}}
				& \textbf{Reduction}
				\\
				\hline
				Equal delays & 1058492 &  753456 & 28.82 \% \\
				different delays  & 1074689 & 763811 & 28.92 \% \\
				prioritizing patterns & 576961 & 439049 & 23.9 \% \\
				
				\hline
			\end{tabular}
			\label{Table:caseStudyReduction}
		\end{center}
	\end{table}
	%-------------------------------------------------
	\section{Related Work}
	
	Our discussion of related work focuses on formal modeling, specification, and verification of interoperable medical systems.
	\cite{Arney-al:ICCPS10} and \cite{6341078} were some of the first works to consider formal modeling and verification of systems based on the ICE architecture.  \cite{6341078} addresses variations of the PCA Monitoring Application in Section~\ref{sec:pca-app}. It focuses on using both UPPAAL's timed automata formalism \cite{bengtsson1995uppaal} and Simulink to capture greater details of the component and system functionality and timing, as well as continuous dynamics of the patient's physiology and its response to the presence of opioids in the bloodstream. Using UPPAAL model checking, various system safety properties are verified including proper halting of PCA infusion when the patient's health is deteriorating. The UPPAAL modeling includes a model of the communication infrastructure which includes capture of non-deterministic error behaviors such as dropped messages.   Our work provides more detailed modeling of communication patterns and component oriented timing specifications, whereas \cite{6341078} provides much greater detail of the medical functionality and patient physiology with continuous dynamics.
	
	The PhD dissertation of Arney \cite{ArneyphDthesis} builds on \cite{Arney-al:ICCPS10} to address expanded versions of the applications in Sections \ref{sec:x-ray-app} and \ref{sec:pca-app}. The approach uses a domain-specific modeling language based on Extended Finite State Machines. A transformation from the modeling language to Java provides simulation capabilities, and a translation to UPPAAL provides model-checking capabilities. Similar to  \cite{Arney-al:ICCPS10}, the focus is on exposing the abstract functional behavior of devices and applications rather than more details of the middleware communication and associated communication timing.
	
	The PhD dissertation of King \cite{KingphDthesis} provides the closest capture of component-related timing properties related to the communication patterns \cite{7318707} and our abstract modeling of the patterns.  \cite{KingphDthesis} defines a domain-specific language for distributed interoperable medical systems with a formal semantics that takes into account the details of tasking and communication. As opposed to focusing on verification, the emphasis of the formalism is to provide a foundation for establishing the soundness of sophisticated real-time scheduling and component interactions of a novel time-partitioning middleware developed by King using Google's OpenFlow software-control network switches. King constructs a dedicated refinement-checking framework that addresses communication time and task and network scheduling using the symbolic representations of timing constraints based on UPPAAL's ``zone" representation.   A number of experiments are performed to assess the scalability and practical effectiveness of the framework.
	
	Larson et al. \cite{Larson-al:NFM2013} specify a more detailed version of the PCA Monitoring Application of Section~\ref{sec:pca-app} using the Architecture Analysis and Definition Language (AADL). Simple functional properties of components are specified on AADL component interfaces using the 
	BLESS interface specification language \cite{Larson-al:NFM2013}. Component behaviors are specified using the BLESS-variant of AADL's Behavior Annex -- a language of concurrent state machines with communication operations based on AADL's real-time port-based communication.   The BLESS theorem prover was used to prove in a compositional manner that component state machine behaviors conform to their BLESS interface specifications and that the composition of components satisfies important system-level behavior specifications. Compared to the approach of this paper,  \cite{Larson-al:NFM2013} focuses on compositional checking of richer functional properties using theorem proving techniques and does not expose the time-related details of communication patterns considered in the model-checking based verification in this paper.
	
	Each of the works above has different strengths that contribute important practical utilities. The long-term vision for specification and verification of interoperable medical systems would almost certainly include a \emph{suite} of techniques that work on a modeling framework supporting realistic and detailed architecture descriptions and embedded system implementations.   Interface specifications would be used to specify component behavior for functional, timing, and fault-related behavior. It is likely that both deductive methods and model checking techniques would be needed to support both compositional contract-based reasoning as well as system state-space exploration (with domain-specific partial order reductions that account for scheduling and atomicity properties of the framework). The work presented in this paper complements the works above by focusing on one part of this larger vision, i.e., it illustrates
	how an existing framework for timed actor-based communication can be leveraged to specify and verify timing-related abstractions of middleware communication between components.
	
	For work that does not focus on MAP-based architectures, 
	\cite{sobrinho2019formal} models and verifies biomedical signal acquisition systems %are embedded systems 
	% that record a patient's signals using electrodes placed on his body. To verify such systems, these systems are modeled in 
	using colored Petri nets in \cite{sobrinho2019formal}. The Model checkers UPPAAL and PRISM \cite{kwiatkowska2011prism} are used to verify autonomous robotic systems as the physical environment of robots has timing constraints and probability \cite{luckcuck2019formal}. To tackle the state-space explosion,  reduction techniques such as symmetry and counter abstraction \cite{clarke2011model} are used to verify the models of swarm robotic systems.
	
	%------------------------------
	
	\section{Conclusion and Future Work}\label{sec::conclude}
	In this paper, we formally modeled composite medical devices interconnected by communication patterns in Timed Rebeca modeling language. We analyzed the configuration of their parameters to assure their timing requirements by Afra tool using the model checking technique. Since modeling many devices using several patterns results in state-space explosion, we proposed a reduction technique by extending FTTS merging technique with regard to the local timing properties. We illustrated the applicability of our approach on two model scenarios inspired by real-world medical systems. %which used three devices and one app communicating by two patterns and 
	We applied our reduction technique on these models. Our results show %almost $30\%$ of 
	significant reduction in systems with a higher number of components. We proposed guidelines and templates for modeling composite systems. Our templates take advantage of inheritance concept in Timed Rebeca in order to have a common communication substrate among instances of patterns. 
	
	Enriching our models by adding some modal behaviors to devices/apps is among our future work. For example,  we can consider different operational modes for the monitoring application in PCA safety interlock application to verify the operational logic of monitoring application. The modes can be normal, degraded, and failed operation. In the normal operation mode, the monitoring app makes decisions based on inputs from both the pulse oximeter and the capnometer. In degraded mode, one of the two devices has failed or gone offline, in that case, the logic in the monitoring app only uses information from the non-failing device. In the failed operation mode, both monitoring devices have gone off-line and the clinician should be notified via an alarm. By modeling the scheduling algorithm of communication network, we can measure communication latency more precisely.

	We aim to generalize our reduction approach by automatically deriving constraints on state variables like the one for $\lastPub$ or message contents to relax shift-equivalence relation in other domains. To this aim, we can use the techniques of static analysis. \FG{Defining a specific language to model the composition and coordination of medical devices, leveraging the proposed communication patterns is among our future work.} %, before moving on to Timed Rebeca.
	
	%-------------------------------------------------
	\section*{Acknowledgment}
	%We would like to thank Ehsan Khamespanh for his kind supports in resolving the problems in using Afra tool. 
	The research of the forth author is partially supported by the KKS Synergy project, SACSys, the SSF project Serendipity, and the KKS Profile project DPAC.

	\bibliographystyle{alphaurl}
	\bibliography{main}
	%\printbibliography
\end{document}